\documentclass[reqno,12pt]{amsart}

\usepackage{amsfonts}
\usepackage{bbm}
\usepackage[comma, authoryear]{natbib}
\usepackage{enumerate}
\usepackage{epsfig}
\usepackage{graphicx}
\usepackage{caption}
\usepackage[utf8]{inputenc}
\usepackage{amsmath}
\usepackage{amssymb}
\usepackage{amsbsy}
\usepackage{mathrsfs}
\usepackage{float}
\usepackage{algorithm}
\usepackage{algpseudocode}
\usepackage{tabularx,booktabs}
\usepackage{mathtools}
\usepackage{subcaption}
\usepackage{cleveref}
\usepackage{todonotes}
\usepackage{enumitem}
\usepackage{multirow}
\numberwithin{equation}{section}
\newcommand{\1}{1}

\def\d{\,\mathrm{d}}
\def\F{\mathcal{F}}
\def\P{\mathbb{P}}
\def\1{\mathbbm 1}
\newcommand{\W}[1]{W^{(#1)}}

\newcolumntype{Y}{>{\centering\arraybackslash}X}

\newcolumntype{b}{Y}
\newcolumntype{s}{>{\hsize=.5\hsize}Y}

\newtheorem{thm}{Theorem}

\newtheorem{prop}{Proposition}

\newtheorem{remark}{Remark}

\newtheorem{myalg}{Algorithm}

\newfloat{subroutine}{htbp}{loa}
\floatname{subroutine}{Subroutine}
\makeatletter\newcommand\l@subroutine{\@dottedtocline{1}{1.5em}{2.3em}}\makeatother

\newfloat{subroutineb}{htbp}{loa}
\floatname{subroutineb}{Subroutine}
\makeatletter\newcommand\l@subroutineb{\@dottedtocline{1}{1.5em}{2.3em}}\makeatother

\renewcommand{\hat}{\widehat}
\newcommand{\pdx}[2]{\frac{\partial #1}{\partial #2}}
\newcommand{\pdxx}[2]{\frac{\partial^2 #1}{\partial #2^2}}
\newcommand{\pdxy}[3]{\frac{\partial^2 #1}{\partial #2 \partial #3}}

\begin{document}
\title[Explicit solution simulation method for the 3/2 model]{Explicit solution simulation method\\ for the 3/2 model}

\author[I.R. Kouarfate]{Iro René Kouarfate}
\address{Department of Mathematics\\
	Universit\'{e} du Qu\'{e}bec \`{a} Montr\'{e}al \\
	Montreal, QC\\
	H3C 3P8 Canada}
\email{kouarfate.iro\_rene@courrier.uqam.ca}

\author[M.A. Kouritzin]{Michael A. Kouritzin}
\address{Department of Mathematical and Statistical Sciences\\
	University of Alberta\\
	Edmonton, AB\\
	T6G 2G1 Canada}
\email{michaelk@ualberta.ca}

\author[A. MacKay]{Anne MacKay}
\address{Department of Mathematics\\
	Universit\'{e} du Qu\'{e}bec \`{a} Montr\'{e}al \\
	Montreal, QC\\
	H3C 3P8 Canada}
\email{mackay.anne@uqam.ca}

\thanks{Partial funding in support of this work
	was provided by an FRQNT grant and by NSERC discovery grants.}
\renewcommand{\subjclassname}{\textup{2010} Mathematics Subject Classification}
\keywords{3/2 model, explicit solutions, weak solutions, stochastic volatility, Monte Carlo simulations, option pricing, non-affine volatility}

\begin{abstract}
	An explicit weak solution for the 3/2 stochastic volatility model is obtained and used to develop a simulation algorithm for option pricing purposes. The 3/2 model is a non-affine stochastic volatility model whose variance process is the inverse of a CIR process. This property is exploited here to obtain an explicit weak solution, similarly to \cite{Kouritzin16} for the Heston model. A simulation algorithm based on this solution is proposed and tested via numerical examples. The performance of the resulting pricing algorithm is comparable to that of other popular simulation algorithms.
\end{abstract}

\maketitle
	

\section{Introduction}\label{intro}

Recent work by \cite{Kouritzin16} shows that it is possible to obtain an explicit weak solution for the Heston model, and that this solution can be used to simulate asset prices efficiently. 
Exploiting the form of the weak solution, which naturally leads to importance sampling, \cite{kouritzinbranching} suggest the use of sequential sampling algorithms to reduce the variance of the estimator, inspired by the particle filtering literature. 
Herein, we show that the main results of \cite{Kouritzin16} can easily be adapted to the 3/2 stochastic volatility model and thus be exploited to develop an efficient simulation algorithm that can be used to price exotic options.

The 3/2 model is a non-affine stochastic volatility model whose analytical tractability was studied in \cite{heston1997simple} and  \cite{lewis2000option}. 
A similar process was used in \cite{ahn1999parametric} to model stochastic interest rates.
Non-affine stochastic volatility models have been shown to provide a good fit to empirical market data, sometimes better than some affine volatility models; see \cite{bakshi2006estimation} and the references provided in the literature review section of \cite{zheng2016pricing}.
The 3/2 model in particular is preferred by \cite{carr2007new} as it naturally emerges from consistency requirements in their proposed framework, which models variance swap rates directly. 

As a result of the empirical evidence in its favor, and because of its analytical tractability, the 3/2 model has gained traction in the academic literature over the past decade.
In particular, \cite{itkin2010pricing} price volatility swaps and options on swaps for a class of Levy models with stochastic time change and use the 3/2 model as a particular case.
The 3/2 model also allows for analytical expressions for the prices of different volatility derivatives; see for example \cite{drimus2012options}, \cite{goard2013stochastic} and \cite{yuen2015pricing}.
\cite{chan2015pricing} consider the 3/2 model for pricing long-dated variance swaps under the real world measure.
\cite{zheng2016pricing} obtain a closed-form partial transform of a relevant density and use it to price variance swaps and timer options.
In \cite{grasselli20174}, the 3/2 model is combined with the Heston model to create the new 4/2 model.

For the 3/2 model's growing popularity, there are very few papers that focus on its simulation. 
One of them is \cite{baldeaux2012exact}, who adapts the method of \cite{BroadieKaya:2006} to the 3/2 model and suggests variance reduction techniques. 
The capacity to simulate price and volatility paths from a given market model is necessary in many situations, from pricing exotic derivatives to developing hedging strategies and assessing risk. 
The relatively small size of the literature concerning the simulation of the 3/2 model could be due to its similarity with the Heston model, which allows for easy transfer of the methods developed for the Heston model to the 3/2 one.
Indeed, the 3/2 model is closely linked to the Heston model; the stochastic process governing the variance of the asset price in the 3/2 is the inverse of a square-root process, that is, the inverse of the variance process under Heston. 

This link between the Heston and the 3/2 model motivates the present work; \cite{Kouritzin16} mentions that his method cannot survive the spot volatility reaching 0.
Since the volatility in the 3/2 model is given by the inverse of a ``Heston volatility'' (that is, the inverse of a square-root process), it is necessary to restrict the volatility parameters in such a way that the Feller condition is met, in order to keep the spot volatility from exploding.
In other words, by definition of the 3/2 model, the variance process always satisfies the Feller condition, which makes it perfectly suitable to the application of the explicit weak solution simulation method of \cite{Kouritzin16}.

It is also worthwhile to note that \cite{kouritzinbranching} notice that the resulting simulation algorithm performs better when the Heston parameters keep the variance process further from 0.
It is reasonable to expect that calibrating the 3/2 model to market data give such parameters, since they would keep the variance process (i.e. the inverse of the Heston variance) from reaching very high values.
This insight further motivates our work, in which we adapt the method of \cite{Kouritzin16} to the 3/2 model.

As stated above, many simulation methods for the Heston model can readily be applied to the 3/2 model.
Most of these methods can be divided into two categories; the first type of simulation schemes relies on discretizing the spot variance and the log-price process.
Such methods are typically fast, but the discretization induces a bias which needs to be addressed, see \cite{lord2010comparison} for a good overview.
\cite{BroadieKaya:2006} proposed an exact simulation scheme which relies on transition density of the variance process and an inversion of the Fourier transform of the integrated variance.
While exact, this method is slow, and has thus prompted several authors to propose approximations and modifications to the original algorithm to speed it up (see for example \cite{andersen2007efficient}).
\cite{begin2015simulating} offers a good review of many existing simulation methods for the Heston model.

The simulation scheme proposed by \cite{Kouritzin16} for the Heston model relies on an explicit weak solution for the stochastic differential equation (SDE) describing the Heston model.
This result leads to a simulation and option pricing algorithm which is akin to importance sampling.
Each path is simulated using an artificial probability measure, called the reference measure, under which exact simulation is possible and fast.
The importance sampling price estimator is calculated under the pricing measure by multiplying the appropriate payoff (a function of the simulated asset price and volatility paths) by a likelihood, which weights each payoff proportionally to the likelihood that the associated path generated from the reference measure could have come from the pricing measure.
The likelihood used as a weight in the importance sampling estimator is a deterministic function of the simulated variance process, and is thus easy to compute.
The resulting pricing algorithm has been shown to be fast and to avoid the problems resulting from discretization of the variance process.

In this paper, we develop a similar method for the 3/2 model by first obtaining a weak explicit solution for the two-dimensional SDE.
We use this solution to develop an option price importance estimator, as well as a simulation and option pricing algorithm.
Our numerical experiments show that our new algorithm performs at least as well as other popular algorithms from the literature. 
We find that the parametrization of the model impacts the performance of the algorithm.

The paper is organized as follows. 
Section 2 contains a detailed presentation of the 3/2 model as well as our main result.
Our pricing algorithm is introduced in Section 3, in which we also outline existing simulation techniques, which we use in our numerical experiments.
The results of these experiments are given in Section 4, and Section 5 concludes.

\section{Setting and main results}

We consider a probability space $(\Omega, \F, \P)$, where $\P$ denotes a pre-determined risk-neutral measure\footnote{Since our goal in this work is to develop pricing algorithms, we only consider the risk-neutral measure used for pricing purposes.} for the 3/2 model. The dynamics of the stock price under this chosen risk-neutral measure are represented by a two-dimensional process $(S,V)=\{(S_t,V_t), t\geq 0\}$ satisfying
\begin{align}
\begin{cases}
\d S_t &= rS_t \d t + \sqrt{V_t} S_t \rho \d \W{1}_t +  \sqrt{V_t} S_t \sqrt{1-\rho^2} \d \W{2}_t\\
\d V_t &= \kappa ~V_t (\theta - V_t) \d t + \varepsilon V^{3/2}_t \d \W{1}_t,
\end{cases}
\label{eq:32_model_SV}
\end{align}
with $S_0 = s_0 > 0$ and $V_0 = v_0 >0$, and where $W = \{(\W{1}_t,\W{2}_t), t\geq 0\}$ is a two-dimensional uncorrelated Brownian motion, $r$, $\kappa$, $\theta$ and $\varepsilon$ are constants satisfying $\kappa > -\frac{\varepsilon^2}{2}$, and $\rho \in [-1,1]$. 
The drift parameter $r$ represents the risk-free rate and $\rho$ represents the correlation between the stock price $S$ and its volatility $V$. 

The restriction $\kappa > -\frac{\varepsilon^2}{2}$ imposed on the parameters keeps the variance process from exploding. 
This property becomes clear when studying the process ${U = \{U_t, t \geq 0\}}$ defined by $U_t = \frac 1{V_t}$ for $t \geq 0$. 
Indeed, it follows from It\^{o}'s lemma that
\begin{align*}
\d U_t &= \kappa \theta \left(\frac{\kappa + \varepsilon^2}{\kappa \theta} U_t\right) ~\d t - \varepsilon \sqrt{U_t} \d \W{1}_t\\
&= \tilde \kappa (\tilde \theta - U_t)  ~\d t + \tilde{\varepsilon} \sqrt{U_t} \d \W{1}_t
\end{align*}
where $\tilde \kappa = \kappa \theta$, $\tilde \theta = \frac{\kappa + \varepsilon^2}{\kappa \theta}$ and $\tilde \varepsilon = -\varepsilon$.
In other words, with the restriction ${\kappa > -\frac{\varepsilon^2}{2}}$, $U$ is a square-root process satisfying the Feller condition $\tilde \kappa \tilde \theta > \frac{\tilde \varepsilon^2}{2}$,  so that ${\P(U_t > 0) = 1}$ for all $t \geq 0$. 

In order to use results obtained for the Heston model and adapt them to the 3/2 model, we express \eqref{eq:32_model_SV} in terms of the inverse of the variance process, $U$, as follows
\begin{align}
\begin{cases}
\d S_t &= rS_t \d t + \sqrt{U_t^{-1}} S_t \rho \d \W{1}_t +  \sqrt{U_t^{-1}} S_t \sqrt{1-\rho^2} \d \W{2}_t\\
\d U_t &= \tilde \kappa (\tilde \theta - U_t)  ~\d t + \tilde{\varepsilon} \sqrt{U_t} \d \W{1}_t,
\end{cases}
\label{eq:32_model_SU}
\end{align}
with $S_0=s_0$ and $U_0 = 1/v_0$.

Although $U$ is a square-root process, \eqref{eq:32_model_SU} is of course not equivalent to the Heston model. Indeed, in the Heston model, it is the diffusion term of $S$, rather than its inverse, that follows a square-root process. 
However, the ideas of Kouritzin (2018) can be exploited to obtain an explicit weak solution to \eqref{eq:32_model_SU}, which will in turn be used to simulate the process.

It is well-known (see for example \cite{hanson2010stochastic}) that if  $n \coloneqq \frac{4\tilde \kappa \tilde \theta}{\tilde\varepsilon^2}$ is a positive integer, the square-root process $U$ is equal in distribution to the sum of $n$ squared Ornstein-Uhlenbeck processes.
Proposition \ref{prop:solution_condition_C} below relies on this result.

\begin{prop}\label{prop:solution_condition_C}
	Suppose that $n = \frac{4\tilde \kappa \tilde \theta}{\tilde\varepsilon^2} \in \mathbb N^+$ and let $\W{2}, Z^{(1)}, \ldots, Z^{(n)}$ be independent standard Brownian motions on some probability space $(\Omega, \F, \P)$.
	For $t\geq 0$, define
	\begin{align*}
	S_t &= s_0 \exp\left\{
		\frac{\rho}{\tilde \varepsilon} \log\left(\frac{U_t}{U_0}\right)
		+\left(r+\frac{\rho \tilde \kappa}{\tilde \varepsilon}\right) t
		\right.\\
		&\qquad\qquad\qquad
		\left.
		-\left(\frac{\rho}{\tilde\varepsilon}
		\left(\tilde \kappa \tilde \theta - \tilde \varepsilon^2/2 \right)+\frac 12\right)\int_0^t U^{-1}_s \d s
		+\sqrt{1-\rho^2}\int_0^t \sqrt{U^{-1}_s} \d W^{(2)}_s 
		\right\},\\
	U_t &= \sum_{i=1}^n \left(Y^{(i)}_t\right)^2,\\
	\intertext{where}
	Y^{(i)}_t &= \frac{\tilde\varepsilon}2 \int_0^t e^{-\frac{\tilde \kappa}2(t-u)} \d Z^{(i)}_u + e^{-\frac{\tilde \kappa}{2} t} Y^{(i)}_0, 
	\qquad \text{with } Y_0 = \sqrt{U_0/n}
	\intertext{and}
	\W{1}_t &= \sum_{i=1}^n \int_0^t \frac{Y^{(i)}_u}{\sqrt{\sum_{j=1}^n (Y^{(j)}_u)^2}} \d Z^{(i)}_u. 
	\end{align*}
	Let $X=(S,U)$, $W=(\W{1},\W{2})$ 
	and let $\{\F_t\}_{t\geq 0}$ be the augmented filtration generated by $(\W{2},Z^{(1)},\ldots,Z^{(n)})$. Then,
	\begin{itemize}
		\item
		$\W{1}$ is a standard Brownian motion, and
		\item
		$(X,W)$, $(\Omega, \F,\P)$, $\{\F_t\}_{t\geq 0}$ is a weak solution to \eqref{eq:32_model_SU}.
	\end{itemize}
\end{prop}

\begin{proof}
We first observe that $Y^{(i)}$, $i \in \{1,\ldots,n\}$, are independent Ornstein-Uhlenbeck processes, and that by L\'{e}vy's characterization, $\W{1}$ is a Brownian motion.
It follows from an application of It\^{o}'s lemma that
\begin{align}
dU_t &= \sum_{i=1}^n \left(\frac{\tilde \varepsilon^2}4 - \tilde \kappa (Y^{(i)}_t)^2\right) \d t 
+ \tilde \varepsilon Y^{(i)}_t \d Z^{(i)}_t\nonumber\\
&= \left(\frac {n\tilde\varepsilon^2}4 - \tilde \kappa \sum_{i=1}^n (Y^{(i)}_t)^2\right) \d t + \tilde \varepsilon \sum_{i=1}^n Y^{(i)}_t \d Z^{(i)}_t\nonumber\\
&= \left(\frac {n\tilde\varepsilon^2}4 - \tilde \kappa U_t\right) \d t + \tilde \varepsilon \sqrt{U_t} \d \W{1}_t,\label{eq:dU_n}
\end{align}  	
where the last equality is obtained by multiplying and dividing the second term on the right-hand side by $\sqrt{\sum_{j=1}^n (Y^{(i)}_t)^2}$.
Here, since we work under the assumption that $n = \frac{4 \tilde\kappa \tilde\theta}{\tilde \epsilon^2}$, \eqref{eq:dU_n} can be re-written as 
\begin{align*}
dU_t &= \tilde \kappa\left(\tilde \theta - U_t\right) \d t + \tilde \varepsilon \sqrt{U_t} \d \W{1}_t.
\end{align*}
An application of It\^{o}'s lemma to $S_t$ completes the proof.
\end{proof}

An alternative, systematic way to verify the functional form for $S_t$ that avoids our It\^{o}-lemma-based guess and verify technique can be found in \cite{Kouritzin16}.

It is likely that for a given market calibration of the 3/2 model, $n=\frac{4\tilde\kappa\tilde\theta}{\tilde \epsilon^2}$ is not an integer. 
For this reason, a more general result is needed to develop a simulation algorithm based on an explicit weak solution. 

We generalize the definition of $n$ and let $n = \max\left(\lfloor \frac{4 \tilde \kappa \tilde \theta}{\tilde \varepsilon^2} + \frac 12 \rfloor,1\right)$. We further define $\tilde \theta_n$ by
\begin{align*}
\tilde\theta_n = \frac{n \tilde \varepsilon^2}{4\tilde\kappa}.
\end{align*}
It follows that $\tilde\kappa \tilde\theta_n = \frac{n\tilde\varepsilon^2}{4}$.

While $U$ above cannot hit $0$ under the Feller condition, it can get arbitrarily close, causing $U^{-1}_t$ to blow up. 
To go beyond the case $\frac{4\tilde\kappa\tilde\theta}{\tilde\epsilon^2} \in \mathbb N$ treated in Proposition \ref{prop:solution_condition_C}, we want to change measures, which is facilitated by stopping $U$ from approaching zero.

This change of measure is needed to readjust the distribution of the paths of $U$ simulated using the (wrong) long-term mean parameter $\tilde \theta_n$ and Proposition \ref{prop:solution_condition_C}.
Indeed, Proposition \ref{prop:solution_condition_C} can be used with $\tilde \theta_n$ since $\frac{4\tilde\kappa\tilde\theta_n}{\tilde\epsilon^2} = n$ is an integer.
Under the new measure, the adjusted paths have the correct distribution, that is, the one associated with the desired parameter $\tilde \theta$.
This idea is made more precise below.

Given a filtered probability space $(\Omega,\F,\{\F_t\}_{t\geq 0},\hat\P)$ with independent Brownian motions $Z^{(1)},\ldots,Z^{(n)}$ and $\W{2}$, and a fixed $\delta >0$, we can define ${(\hat S, \hat U) = \{(\hat S_t,\hat U_t)\}_{t\geq0}}$ by
\begin{align}
\hat S_t &= 
s_0 \exp\left\{\frac {\rho}{\tilde \varepsilon}\log (\hat U_t / U_0) 
+ \left(r + \frac{\rho \tilde \kappa}{\tilde \varepsilon}\right)t\right.\nonumber\\
&\qquad\qquad
-\left.\left(\frac {\rho}{\tilde \varepsilon} \left(\tilde\kappa \tilde \theta - \tilde \varepsilon^2/2\right) + \frac 12\right) \int_0^t \hat U^{-1}_s \d s
+ \sqrt{1-\rho^2} \int_0^t \hat U^{-1/2}_s \d \W{2}_s
\right\},\label{eq:hat_S}\\
\hat U_t &= \sum_{i=1}^n (Y^{(i)}_t)^2,\label{eq:hat_U}\\
\intertext{and $\tau_\delta = \inf\{t \geq 0: \hat U_t \leq \delta\}$, where}
Y^{(i)}_t &= \frac{\tilde\varepsilon}2 \int_0^{t \wedge \tau_\delta} e^{-\frac{\tilde \kappa}2(t-u)} \d Z^{(i)}_u + e^{-\frac{\tilde \kappa}{2} (t \wedge \tau_\delta)} Y^{(i)}_0, 
\qquad \text{with } Y_0 = \sqrt{U_0/n}
\label{eq:Y}
\end{align}
for $i \in \{1,\ldots,n\}$.

Theorem \ref{thm:solution}, to follow immediately, shows that it is possible to construct a probability measure on $(\Omega,\F)$ under which $(\hat S, \hat U)$ satisfies \eqref{eq:32_model_SU} until $\hat U$ drops below a pre-determined threshold $\delta$.

\begin{thm}\label{thm:solution}
	Let $(\Omega,\F,\{\F_t\}_{t\geq 0},\hat\P)$ be a filtered probability space on which $Z^{(1)},\ldots,Z^{(n)}, \W{2}$ are independent Brownian motions. Let $(\hat S, \hat U)$ be defined as in \eqref{eq:hat_S} and \eqref{eq:hat_U} and let $\tau_\delta = \inf\{ t\geq 0: \hat U_t \leq \delta\}$ for some $\delta \in (0,1)$. Define
	\begin{align}
	\hat L^{(\delta)}_t &= \exp\left\{
	- \frac{\tilde \kappa (\tilde \theta_n - \tilde \theta)}{\tilde\varepsilon}\int_0^t \hat U_v^{-1/2} \d \hat W^{(1)}_v - \frac {\tilde\kappa^2}2 \left(\frac{\tilde \theta_n - \tilde \theta}{\tilde\varepsilon}\right)^2\int_0^t \hat U^{-1}_v \d v
	\right\}
	\label{eq:hat_L}
	\intertext{with}
	\hat W^{(1)}_t &= \sum_{i=1}^n \int_0^t \frac{Y^{(i)}_u}{\sqrt{\sum_{j=1}^n (Y^{(j)}_u)^2}} \d Z^{(i)}_u
	\label{eq:hat_W1}
	\end{align}
	and $\P^{\delta}(A) = \hat E[1_A \hat L^{(\delta)}_{T}] ~\forall A \in \F_T$ for $T>0$.
	
	
	Then, under the probability measure $\P^{\delta}$, $(\W{1},\W{2})$, where
	\begin{align*}
	\W{1}_t = \hat W^{(1)}_t + \tilde\kappa \frac{\tilde\theta - \tilde \theta_n}{\tilde \varepsilon} \int_0^{t\wedge\tau_\delta} \hat U^{-1/2}_s \d s,
	\end{align*}
	are independent Brownian motions and $(\hat S, \hat U)$ satisfies
	\begin{align}
	\begin{split}
	\d\hat S_t &= 
	\begin{cases}
	r\hat S_t \d t + \hat U_t^{-1/2} \hat S_t \rho \d W^{(1)}_t +  \hat U_t^{-1/2} \hat S_t \sqrt{1-\rho^2} \d W^{(2)}_t, &t \leq \tau_\delta\\
	r_\delta \hat S_t \d t + \sigma_\delta \hat S_t \d W^{(2)}_t,
	&t > \tau_\delta,
	\end{cases}\\
	\d\hat U_t &= 
	\begin{cases}
	\tilde \kappa (\tilde \theta - \hat U_t)  ~\d t + \tilde{\varepsilon} \hat U_t^{1/2} \d W^{(1)}_t, &t \leq \tau_\delta\\
	0, &t > \tau_\delta
	\end{cases}
	\end{split}\label{eq:solution}
	\end{align}
	on $[0,T]$, with
	\begin{align*}
	r_\delta &= r+\frac{\rho}{2\tilde\varepsilon\delta}
	\left(2\tilde\kappa\delta - 2\tilde\kappa\tilde\theta + \tilde\varepsilon^2 - \rho\tilde\varepsilon^2\right),\\
	\sigma_\delta &= \sqrt{\frac{1-\rho^2}{\delta}}.
	\end{align*}
\end{thm}

\begin{proof}
	Let $D=\mathcal S(\mathbb R^2)$, the rapidly decreasing functions. They separate points and are closed under multiplication so they separate Borel probability measures (see \cite{blount2010convergence}) and hence are a reasonable martingale problem domain.
	
	To show that $(\hat X, W)$, $(\Omega,\F,\P^{\delta})$, $\{\hat \F_t\}_{t\geq 0}$, with $\hat X = (\hat S, \hat U)$ and $W = (W^{(1)},W^{(2)})$, is a solution to \eqref{eq:solution}, we show that it solves the martingale problem associated with the linear operator 
	\begin{align*}
	\mathcal A_t f(s,u)&=
	\left(rs f_s(s,u) +
	\tilde \kappa (\tilde \theta - u) f_u(s,u) 
	+ \frac 12 s^2 u^{-1} f_{ss}(s,u)
	+ \rho \tilde\varepsilon s f_{su}(s,u)
	\right.\\
	&\qquad\left.
	+ \frac 12 \varepsilon^2 u f_{uu}(s,u)
	\right) \1_{[0,\tau_\delta]}(t)
	+\left(r_\delta s f_s(s,u) + \frac {1-\rho^2}{2\delta^2} f_{ss} s^2\right)\1_{[\tau_\delta,T]}(t)
	\end{align*}		
	where $f_{s} = \pdx{f(s,u)}{s}$, $f_{u} = \pdx{f(s,u)}{u}$, $f_{ss} = \pdxx{f(s,u)}{s}$, $f_{uu} = \pdxx{f(s,u)}{u}$ and $f_{su} = \pdxy{f(s,u)}{s}{u}$.
	That is, we show that for any function $f \in D$, the process
	\begin{align*}
	M_t(f) = f(\hat S_t, \hat U_t) - f(\hat S_0, \hat U_0) - \int_0^t (\mathcal A_s f)(\hat S_v, \hat U_v) \d v,
	\end{align*}
	is a continuous, local martingale. 
	
	First, we note by \eqref{eq:hat_S}, \eqref{eq:hat_U}, \eqref{eq:Y} as well as It\^{o}'s lemma that $(\hat S, \hat U)$ satisfies a two-dimensional SDE similar to the 3/2 model \eqref{eq:32_model_SU}, but with parameters $\kappa$, $\theta_n$, $r_\delta$ and $\hat r_t = r - \frac{\tilde\kappa\rho}{\tilde \varepsilon}(\tilde \theta - \tilde \theta_n)\hat U^{-1}_t$.
	That is, $((\hat S, \hat U),\hat W)$, $(\Omega, \F, \hat \P)$, $\{\hat \F_t\}_{t\geq0}$, where $\{\hat \F_t\}_{t\geq0}$ is the augmented filtration generated by $(Z_1,\ldots,Z_n,\W{2})$, is a solution to
	\begin{align}
	\d \hat S_t &= 
	\begin{cases}
	\hat r_t \hat S_t \d t + \hat U_t^{-1/2} \hat S_t \rho \d \hat W^{(1)}_t +  \hat U_t^{-1/2}\hat S_t \sqrt{1-\rho^2} \d \W{2}_t,
	& t \leq \tau_\delta,\\
	r_\delta \hat S_t \d t + \sigma_\delta \hat S_t\d W^{(2)}_t,
	&t > \tau_\delta,
	\end{cases}\\
	\d\hat U_t &=
	\begin{cases}
	\tilde \kappa (\tilde \theta_n - \hat U_t)  ~\d t + \tilde{\varepsilon} \hat U_t^{1/2} \d \hat W^{(1)}_t, 
	& t \leq \tau_\delta, \\
	0, &t > \tau_\delta.
	\end{cases}
	\label{eq:dSU_hat}
	\end{align}
	with $\hat S_0=s_0$, $\hat U_0 = 1/v_0$ and $\hat W^{(1)}$ defined by \eqref{eq:hat_W1}.
	It follows that for any function $f \in D$,
	\begin{align}
	\begin{split}
	df(\hat S_t, \hat U_t) &= 
	\mathcal L_t f(\hat S_t, \hat U_t) dt\\
	&\qquad +  
	\left( \rho \hat S_t \hat U^{-1/2}_t f_s(\hat S_t,\hat U_t)
	+ \tilde \varepsilon \hat U^{1/2}_t f_u(\hat S_t,\hat U_t) \right) \1_{[0,\tau_\delta]}(t) \d \hat W^{(1)}_t \\
	&\qquad + 
	\left(\hat U^{-1/2}_t \1_{[0,\tau_\delta]}(t) + \delta^{-1/2} \1_{[\tau_\delta,T]}\right) \sqrt{1-\rho^2} \hat S_t f_s(\hat S_t,\hat U_t) \d \W{2}_t,
	\end{split}
	\label{eq:df_SU}
	\end{align}
	where the linear operator $\mathcal L$ is defined by
	\begin{align}
	\begin{split}
	\mathcal L_t f(s,u) &= 
	\left(\hat r_t s f_s(s,u) +
	\tilde \kappa (\tilde \theta_n - u) f_u(s,u) 
	+ \frac 12 s^2 u^{-1} f_{ss}(s,u)
	+ \rho \tilde\varepsilon s f_{su}(s,u)
	\right.\\
	&\qquad\left.
	+ \frac 12 \varepsilon^2 u f_{uu}(s,u)
	\right) \1_{[0,\tau_\delta]}(t)
	+\left(r_\delta s f_s(s,u) + \frac {1-\rho^2}{2\delta^2} f_{ss} s^2\right)\1_{[\tau_\delta,T]}(t).
	\end{split}
	\label{eq:Lf}
	\end{align}
	We observe that $\hat L^{(\delta)}_t$ satisfies the Novikov condition, since by definition of $\hat U_t$, 
	\begin{align*}
		\frac{|\tilde \kappa (\tilde \theta_n - \tilde \theta)|^2}{\tilde \varepsilon^2 \hat U_t} \leq 	\frac{|\tilde \kappa (\tilde \theta_n - \tilde \theta)|^2}{\tilde \varepsilon^2 \delta},
	\end{align*}
	$\hat \P$-a.s. for all $t \geq 0$. 
	It follows that $\hat L^{(\delta)}_t$ is a martingale and that $\P^{\delta}$ is a probability measure.
	
	We also have from \eqref{eq:hat_L} and \eqref{eq:df_SU}  that for $f(s,u) \in C^2([0,\infty]^2)$, 
	\begin{align}
	\left[\hat L^{(\delta)},f(\hat S, \hat U)\right]_t
	&= \int_0^{t \wedge \tau_\delta} \hat L^{(\delta)}_v
	\left( (r - \hat r_v) \hat S_v f_s(\hat S_v, \hat U_v)
	+ \tilde \kappa(\tilde \theta - \tilde \theta_n) f_u(\hat S_v, \hat U_v)
	\right) \d v.
	\label{eq:Lf_qv}
	\end{align}
	 
	Next we define the process $\hat M(f)$ for any $f \in D$ by
	\begin{align}
	\hat M_t(f) &= \hat L^{(\delta)}_t f(\hat S_t, \hat U_t) - \hat L^{(\delta)}_0 f(\hat S_0, \hat U_0) - \int_0^t \hat L^{(\delta)}_v \mathcal A_v f(\hat S_v, \hat U_v) \d v\nonumber\\
	&= \hat L^{(\delta)}_t f(\hat S_t, \hat U_t) - \hat L^{(\delta)}_0 f(\hat S_0, \hat U_0) -\left[\hat L^{(\delta)},f(\hat S, \hat U)\right]_t - \int_0^t \hat L^{(\delta)}_v \mathcal L^{(\delta)}_v f(\hat S_v, \hat U_v) \d v.
	\label{eq:Mhat_dv}
	\end{align}
	Using integration by parts, we obtain 
	\begin{align*}
	\hat M_t(f) &= 
	\int_0^t \hat L^{(\delta)}_v \d f(\hat S_v,\hat U_v) \d v
	+ \int_0^t f(\hat S_v, \hat U_v) \d \hat L^{(\delta)}_v
	-\int_0^t L^{(\delta)}_v \mathcal L^{(\delta)}_v f(\hat S_v, \hat U_v) \d v\\
	&= \int_0^t \hat L^{(\delta)}_v \left[
	\tilde \kappa (\tilde \theta - \tilde \theta_n) \hat U^{-1/2}_v f(\hat S_v, \hat U_v) 
	+ \left(\rho \hat S_v \hat U^{-1/2}_v f_s(\hat S_v, \hat U_v) \right.\right.\\
	&\qquad \left.\left. + \tilde \varepsilon \hat U^{1/2}_v f_u(\hat S_v, \hat U_v)\right) \1_{[0,\tau_\delta]}(v) \right] \d \hat W^{(1)}_v\\
	&\qquad
	+ \int_0^t \hat L^{(\delta)}_v \left(\hat U^{-1/2}_v \1_{[0,\tau_\delta]}(v) + \delta^{-1/2} \1_{[\tau_\delta,T]}(v)\right) \sqrt{1-\rho^2} \hat S_v f_s(\hat S_v,\hat U_v) \d \W{2}_v
	\end{align*}
	so $\hat M_t(f)$ is a local martingale. 
	However, since $f$ is rapidly decreasing, $s f_s(s,u)$, $u f_u(s,u)$, $s f_{su}(s,u)$ and $u f_{uu}(s,u)$ are all bounded.
	We also have that $\hat U_v \geq \delta$ and $\hat L^{(\delta)}_v$ is integrable for all $v$.
	Hence, it follows by \eqref{eq:Lf}, \eqref{eq:Lf_qv}, \eqref{eq:Mhat_dv} and Tonelli that $\hat M(f)$ is a martingale.
	
	To finish the proof, it suffices to follow the remark on p.174 of \cite{ethiermarkov} and show that
	\begin{align}
		E\left[\left(f(\hat S_{t_{n+1}}, \hat U_{t_{n+1}}) - f(\hat S_{t_n}, \hat U_{t_n}) 
			-\int_{t_n}^{t_{n+1}} \mathcal A_vf(\hat S_v, \hat U_v) \d v\right)\prod_{k=1}^n h_k(\hat S_{t_k}, \hat U_{t_k})\right] = 0,
			\label{eq:EKequality}
	\end{align}
	for $0 \leq t_1 < t_2 < \ldots < t_{n+1}$, $f \in D$, $h \in B(\mathbb R^2)$ (the bounded, measurable functions) and where $\hat E[\cdot]$ denotes the $\hat \P$-expectation. To do so, we re-write the left-hand side of \eqref{eq:EKequality} as
	\begin{align*}
	&\hat E\left[\hat L^{(\delta)}_{t_{n+1}} \left(f(\hat S_{t_{n+1}}, \hat U_{t_{n+1}}) - f(\hat S_{t_n}, \hat U_{t_n}) 
	-\int_{t_n}^{t_{n+1}} \mathcal A_vf(\hat S_v, \hat U_v) \d v\right)\prod_{k=1}^n h_k(\hat S_{t_k}, \hat U_{t_k})\right]\\
	&=\hat E\left[ \left(
		\hat L^{(\delta)}_{t_{n+1}} f(\hat S_{t_{n+1}}, \hat U_{t_{n+1}}) 
		- \hat L^{(\delta)}_{t_{n}} f(\hat S_{t_n}, \hat U_{t_n}) 
		-\int_{t_n}^{t_{n+1}} \hat L^{(\delta)}_{v} \mathcal A_vf(\hat S_v, \hat U_v) \d v\right)\prod_{k=1}^n h_k(\hat S_{t_k}, \hat U_{t_k})\right]\\
	&=\hat E\left[\left(\hat M_{t_{n+1}}(f) - \hat M_{t_n}(f)\right)\prod_{k=1}^n h_k(\hat S_{t_k}, \hat U_{t_k})\right],
	\end{align*}
	which is equal to 0 since $\hat M(f)$ is a martingale. We can then conclude that $(\hat S, \hat U)$ solves the martingale problem for $\mathcal A$ with respect to $\hat \P$.
\end{proof}

\begin{remark}
	In Theorem \ref{thm:solution}, we indicate the dependence of the process $\hat L^{(\delta)}$ on the threshold $\delta$ via the superscript. Indeed, $\hat L^{(\delta)}$ depends on $\delta$ through $\hat U$. Going forward, for notational convenience, we drop the superscript, keeping in mind the dependence of the likelihood process on $\delta$.
\end{remark}

\section{Pricing algorithm}

In this section, we show how Theorem \ref{thm:solution} can be exploited to price a financial option in the 3/2 model. 
First, we justify that $(\hat S, \hat U)$ defined in \eqref{eq:hat_S} and \eqref{eq:hat_U} can be used to price an option in the 3/2 model, even if they satisfy \eqref{eq:32_model_SU} only up to $\tau_\delta$. 
We also present an algorithm to simulate paths of $(\hat S,\hat U)$ under the 3/2 model as well as the associated importance sampling estimator for the price of the option.

\subsection{Importance sampling estimator of the option price}

For the rest of this paper, we consider an option with maturity $T \in \mathbb R^+$ whose payoff can depend on the whole path of $\{(S_t,V_t)\}_{t \in [0,T]}$, or equivalently, $\{(S_t,U_t)\}_{t \in [0,T]}$. 
Indeed, since $V_t = U_t^{-1}$ for all $0 \leq t \leq T$ and to simplify exposition, we will keep on working in terms of $U$, the inverse of the variance process, going forward. 
We consider a \emph{payoff function} $\phi_T(S,U)$ with $E[|\phi_T(S,U)|] < \infty$.
We call $\pi_0 = E[\phi_T(S,U)]$ the \emph{price of the option} and the function $\phi_T$, its discounted payoff.
For example, a call option, which pays out the difference between the stock price at maturity, $S_T$, and a pre-determined exercise price $K$ if this difference is positive, has discounted payoff function $e^{-rT}\max(S_T-K,0)$ and price $E[e^{-rT}\max(S_T-K,0)]$.

\begin{remark}
	We work on a finite time horizon and the option payoff function $\phi_T$ only depends on $(S,U)$ up to $T$. We use the index $T$ to indicate this restriction on $(S,U)$.
\end{remark}

The next proposition shows that it is possible to use $(\hat S, \hat U)$, rather than $(S,U)$, to price an option in the 3/2 model.

\begin{prop}\label{prop:convergence}
	Suppose $(S,U)$ is a solution to the 3/2 model \eqref{eq:32_model_SU} on probaiblity space $(\Omega, \mathcal F, \P)$ and $\tau_\delta = \inf\{t \geq 0: U_t \leq \delta\}$. Define  $(\hat S, \hat U)$ by \eqref{eq:hat_S} and \eqref{eq:hat_U}, set $\hat \tau_\delta = \inf\{t \geq 0: \hat U_t \leq \delta\}$ for $\delta \in (0,1)$ and let $\phi_T(S,U)$ be a payoff function satisfying $E[|\phi_T(S,U)|]<\infty$. Then, 
	\begin{align*}
	\lim_{n\rightarrow \infty} E^{1/n}[\phi_T(\hat S, \hat U) \1_{\{\tau_{1/n}>T\}}] = E[\phi_T(S,U)],
	\end{align*}
	where $E^\delta[\cdot]$ denotes the expectation under the measure $\P^\delta$ defined in Theorem \ref{thm:solution}.
\end{prop}

\begin{proof}
	By Theorem \ref{thm:solution}, $(\hat S, \hat U)$ satisfies \eqref{eq:32_model_SU} on $[0,\tau_{1/n}]$ under the measure $\P^{1/n}$. It follows that
	\begin{align*}
	E^{1/n}[\phi_T(\hat S, \hat U) \1_{\{\hat \tau_{1/n}>T\}}] = E[\phi_T(S, U) \1_{\{\tau_{1/n}>T\}}].
	\end{align*}
	Because $U$ satisfies the Feller condition, $\lim_{n\rightarrow \infty} \1_{\{\tau_{1/n} \leq T\}} = 0, ~\P$-a.s. and
	\begin{align*}
	\lim_{n\rightarrow \infty} E^{1/n}[\phi_T(\hat S, \hat U) \1_{\{\hat \tau_{1/n}>T\}}]
	= \lim_{n\rightarrow \infty} E[\phi_T(S, U) \1_{\{\tau_{1/n}>T\}}]
	= E[\phi_T(S,U)]
	\end{align*}
	by the dominated convergence theorem.
	
\end{proof}

We interpret Proposition \ref{prop:convergence} in the following manner: by choosing $\delta$ small enough, it is possible to approximate $\pi_0$ by $\pi^{(\delta)}_0 \coloneqq E^{\delta}[\phi_T(\hat S, \hat U) \1_{\{\tau_{\delta}>T\}}]$, that is, using $(\hat S, \hat U)$ rather than $(S,U)$.
The advantage of estimating the price of an option via $(\hat S, \hat U)$ is that the trajectories can easily be simulated exactly under the reference measure $\hat \P$ defined in Theorem \ref{thm:solution}.
In practice, we will show in Section \ref{sec:NumericalResults} that for reasonable 3/2 model calibrations, it is usually possible to find $\delta$ small enough that $E^{\delta}[\phi_T(\hat S, \hat U) \1_{\{\tau_{\delta}>T\}}]$ is almost undistinguishable from $\pi_0$.

In the rest of this section, we explain how $\pi^{(\delta)}_0$ can be approximated with Monte Carlo simulation.
As mentioned above, paths of $(\hat S, \hat U)$ are easily simulated under the reference measure $\hat \P$, not under $\P^{\delta}$. 
It is therefore necessary to express $\pi^{(\delta)}_0$ using Theorem \ref{thm:solution} in the following manner
\begin{align}
\pi^{(\delta)}_0 
= E^{\delta}[\phi_T(\hat S, \hat U) \1_{\{\tau_{\delta}>T\}}] 
= \hat E[\hat L_T~ \phi_T(\hat S, \hat U) \1_{\{\tau_{\delta}>T\}}].
\label{eq:pi_delta}
\end{align}

From \eqref{eq:pi_delta} and the strong law of large numbers, we can define $\hat \pi_0^{(\delta)}$, an importance estimator for $\pi_0^{(\delta)}$, by
\begin{align}
	\hat \pi_0^{(\delta)} = \frac{\sum_{j=1}^N \phi_T(\hat S^{(j)}, \hat U^{(j)}) \hat L^{(j)}_T \1_{\{\tau^{(j)}_\delta > T\}}}{\sum_{j=1}^N \hat L^{(j)}_T},
	\label{eq:hat_pi}
\end{align}

where $\left\{\hat S^{(j)}, \hat U^{(j)}, \hat L^{(j)}\right\}_{j=1}^N$ are $N \in \mathbb N$ simulated paths of $(\hat S, \hat U, \hat L)$.

\subsection{Simulating sample paths}

In light of Proposition \ref{prop:convergence}, we now focus on the simulation of $(\hat S_t, \hat U_t, \hat L_t)_{t \leq \tau_\delta}$.
Using \eqref{eq:hat_S} and \eqref{eq:Y}, $\hat S$ and $Y$ can easily be discretized for simulation purposes. 
To simplify the simulation of the process $\hat L$, we write \eqref{eq:hat_L} as a deterministic function of $\hat U$ in Proposition \ref{prop:Lt} below.

\begin{prop}\label{prop:Lt}
	Let $\hat L_t$ be defined as in Theorem \ref{thm:solution}, with $\hat U$ defined by \eqref{eq:hat_U}. Then, for $t \leq \tau_\delta$, $\hat L_t$ can be written as
	\begin{align}
	\hat L_t = \exp\left\{
	\frac{-(\tilde\kappa \tilde\theta_n - \tilde\kappa \tilde\theta)}{\tilde \varepsilon^2} 
	\left[
	\log(\hat U_t / \hat U_0)
	+\tilde \kappa t
	+\frac{\tilde\kappa \tilde\theta - 3\tilde\kappa \tilde\theta_n +\tilde\varepsilon^2}{2}
	\int_0^t \hat U_s^{-1}~\d s
	\right]
	\right\}.
	\label{eq:hat_L_alt}
	\end{align}
\end{prop}

\begin{proof}
	An application of It\^o's lemma to $\log \hat U_t$ for $t \leq \tau_\delta$ yields
	\begin{align}
	\log(\hat U_t / \hat U_0) = (\tilde\kappa \tilde\theta_n - \tilde \varepsilon^2/2) \int_0^t \hat U_s^{-1} \d s 
	-\tilde\kappa t
	+ \tilde \varepsilon \int_0^t \hat U_s^{-1/2} \d \hat W^{(1)}_s.
	\label{eq:proof_Lt}
	\end{align}
	Isolating $\int_0^t \hat U_s^{-1/2} \d \hat W^{(1)}_s$ in \eqref{eq:proof_Lt} and replacing the resulting expression in \eqref{eq:hat_L} gives the result.	
	
\end{proof}

For $t \in [0,T)$ and $h \in (0,T-t)$, for simulation purposes, we can re-write \eqref{eq:hat_S}, \eqref{eq:Y} and \eqref{eq:hat_L_alt} in a recursive manner as
\begin{align}
\hat S_{t+h} &= \hat S_t \exp\left\{
	\frac{\rho}{\tilde\varepsilon} \log(\hat U_{t+h}/ \hat U_t)
	+ ah\nonumber
	- b \int_t^{t+h} \hat U^{-1}_s \d s\right.\nonumber\\
	&\qquad\qquad\qquad\qquad\qquad\qquad\left. 
	+\sqrt{1-\rho^2} \int_t^{t+h} \hat U_s^{-1/2} \d W^{(2)}_s,
\right\}\label{eq:S_th}\\
Y^{(i)}_{t+h} &= Y^{(i)}_t e^{-\frac{\tilde\kappa}{2}h} + \frac{\tilde\varepsilon}{2}\int_t^{(t+h)\wedge \tau_\delta} e^{-\frac{\tilde\kappa}{2}(t+h-u)} \d Z_u, 
\qquad \text{for }i = 1,\ldots,n,\label{eq:Y_th}\\
\intertext{and}
\hat L_{(t+h)\wedge \tau_\delta} &= L_t \exp\left\{c\left(\log(\hat U_{(t+h)\wedge \tau_\delta}/ \hat U_t) + \tilde \kappa (h\vee (\tau_\delta-t))\right.\right. \nonumber\\
&\qquad\qquad\qquad\qquad\qquad\qquad\qquad\qquad \left.\left.+ d \int_t^{(t+h)\wedge \tau_\delta} \hat U^{-1}_s \d s\right)\right\},
\label{eq:L_th}
\end{align}
where
\begin{align*}
a = r + \frac{\rho \tilde\kappa}{\tilde \varepsilon} \qquad
b = \frac{\rho}{\tilde\varepsilon}(\tilde\kappa \tilde\theta - \tilde \varepsilon^2/2) \qquad
c = -\frac{\tilde\kappa\tilde\theta_n - \tilde\kappa\tilde\theta}{\varepsilon^2} \qquad
d = \frac{\tilde\kappa\tilde\theta - 3\tilde\kappa\tilde\theta_n + \tilde \varepsilon^2}{2}.
\end{align*}

We now discuss the simulation of $(\hat S_{t+h}, \hat U_{t+h}, \hat L_{t+h})$ given $(\hat S_{t}, \hat U_{t}, \hat L_{t})$, as well as $\{Y^{(i)}_t\}_{i=1}^n$.
Typically, $h$ will be a small time interval, that is, we consider $h \ll T$.
It is easy to see from the above that given $Y^{(i)}_t$, $Y^{(i)}_{t+h}$ follows a Normal distribution with mean $Y^{(i)}_t e^{-\frac{\tilde\kappa}{2}h}$ and variance $\frac{\tilde \varepsilon^2}{4\tilde\kappa} (1-e^{-h\tilde\kappa})$. 
The simulation of $Y^{(i)}_{t+h}$ given $Y^{(i)}_t$ is thus straightforward. 
$\hat U_{t+h}$ can then be obtained by \eqref{eq:hat_U} as the sum of the squares of each $Y^{(i)}_{t+h}$, for $i = 1, \ldots, n$.

Given simulated values $\hat U_{t+h}$ and $\hat U_{t}$, the term $\int_t^{t+h} \hat U^{-1}_s \d s$, which appears in both $\hat S_{t+h}$ and $\hat L_{t+h}$, can be approximated using the trapezoidal rule by letting
\begin{align}
	\int_t^{t+h} \hat U^{-1}_s \d s \approx 
	\frac{(\hat U^{-1}_{t} + \hat U^{-1}_{t+h})}{2} h.
	\label{eq:int_U}
\end{align}
More precise approximations to this integral can be obtained by simulating intermediate values $\hat U^{-1}_{t+ih}$ for $i \in (0,1)$ and using other quadrature rules.
In \cite{Kouritzin16} and \cite{kouritzinbranching}, Simpson's $\frac{1}{3}$ rule was preferred.
In this section, we use a trapezoidal rule only to simplify the exposition of the simulation algorithm.

Given that $\hat U_{t+h} > \delta$ and once an approximation for the deterministic integral $\int_t^{t+h} \hat U^{-1}_s \d s$ is calculated, $\hat L_{t+h}$ can be simulated using \eqref{eq:L_th}.
To generate a value for $\hat S_{t+h}$, it suffices to observe that conditionally on $\{\hat U_s\}_{s \in [t,t+h]}$, $\int_t^{t+h} \hat U_s^{-1/2}~\d W^{(2)}_s$ follows a Normal distribution with mean 0 and variance $\int_t^{t+h} \hat U^{-1}_s \d s$.

The resulting algorithm produces $N$ paths of $(\hat S, \hat U, \hat L)$ and the stopping times $\tau_\delta$ associated with each path; it is presented in Algorithm \ref{algo:simulation}, in the appendix.
These simulated values are then used in \eqref{eq:hat_pi} to obtain an estimate for the price of an option.

\section{Numerical experiment}\label{sec:NumericalResults}

\subsection{Methods and parameters}

In this section, we assess the performance of the pricing algorithm derived from Theorem \ref{thm:solution}.
To do so, we use Monte Carlo simulations to estimate the price of European call options.
These Monte Carlo estimates are compared with the exact price of the option, calculated with the analytical expression available for vanilla options in the 3/2 model (see for example \cite{lewis2000option} and \cite{carr2007new}).
More precisely, we consider the discounted payoff function $\phi_T(S,U) = e^{-rT}\max(S_T-K,0)$ for $K>0$ representing the exercise price of the option and we compute the price estimate according to \eqref{eq:hat_pi}.

The precision of the simulation algorithm is assessed using either the mean square error or the relative mean square error, as indicated. We define the mean square error by
\begin{align*}
MSE = E[(\pi_0 - \hat \pi_0^{(\delta)})^2]
\end{align*}
and the relative mean square error by
\begin{align*}
RelMSE = \frac{E[(\pi_0 - \hat \pi_0^{(\delta)})^2]}{\pi_0},
\end{align*}
where $\pi_0$ is the exact price of the option and $\hat \pi_0^{(\delta)}$ is the estimate calculated with \eqref{eq:hat_pi}.
The expectations above are approximated by calculating the estimates a large number of times and taking the mean over all runs.

Throughout this section, we consider the five parameter sets presented in Table \ref{tab:parameters}. 
Parameter set 1 (PS1) was used in \cite{baldeaux2012exact}.
Parameter set 2 (PS2) was obtained by \cite{drimus2012options} via the simultaneous fit of the 3/2 model to 3-month and 6-month S\&P500 implied volatilities on July 31, 2009.
The three other parameter sets are modifications of PS2: PS3 was chosen so that $\frac{4\tilde\kappa\tilde\theta}{\tilde\varepsilon^2} \in \mathbb N$, and PS4 and PS5 were selected to have a higher $n$.
Recalling that $n = \max\left(\lfloor \frac{4 \tilde \kappa \tilde \theta}{\tilde \varepsilon^2} + \frac 12 \rfloor,1\right)$ represents the number of Ornstein-Uhlenbeck processes necessary to simulate the variance process, we have that $n=204$ for PS1, $n=5$ for PS2 and PS3 and $n=12$ for PS4 and PS5.

Throughout the numerical experiments, the threshold we use is $\delta = 10^{-5}$.
For all parameter sets, the simulated process $U$ never crossed below this threshold. 
Therefore, any $\delta$ below $10^{-5}$ would have yielded the same results.

\begin{table}[H]
\begin{center}
	\caption{Parameter sets}
	\label{tab:parameters}
\begin{tabular}{ccccccccc}
\hline
& $S_0$ & $V_0$ & $\kappa$ & $\theta$ &	$\varepsilon$ & $\rho$ & $r$ & $\boldsymbol{4\tilde\kappa\tilde\theta / \tilde\varepsilon^2}$\\	
\hline
PS1 & 1 & 1 & 2 & $1.5$ & $0.2$ & $-0.5$ & $0.05$ & $\mathbf{204}$ \\
PS2 & 100 & $0.06$ & $22.84$ & $0.218$ & $8.56$ & $-0.99$ & $0.00$ & $\mathbf{5.25}$\\
PS3 & 100 & $0.06$ & $18.32$ & $0.218$ & $8.56$ & $-0.99$ & $0.00$ &$\mathbf{5.00}$\\
PS4 & 100 & $0.06$ & $19.76$ & $0.218$ & $3.20$ & $-0.99$ & $0.00$ &$\mathbf{11.72}$\\
PS5 & 100 & $0.06$ & $20.48$ & $0.218$ & $3.20$ & $-0.99$ & $0.00$ &$\mathbf{12.00}$\\
\hline
\end{tabular}
\end{center}
\end{table}

\subsection{Results}

In this section, we present the results of our numerical experiments. We first test the sensitivity of our simulation algorithm to $n$, the number of Ornstein-Uhlenbeck processes to simulate.
We then compare the performance of our algorithm to other popular ones in the literature.

\subsubsection{Sensitivity to $n$}

We first test the impact of $n$ on the precision of the algorithm. 
Such an impact was observed in \cite{kouritzinbranching} in the context of the Heston model.
To verify whether this also holds for the 3/2 model, we consider the first three parameter sets and price at-the-money (that is, $K=S_0$) European call options.
For PS1, we follow \cite{baldeaux2012exact} and compute the price of a call option with maturity $T=1$.
The exact price of this option is 0.4431.
PS2 and PS3 are used to obtain the price of at-the-money call options with $T=0.5$, with respective exact prices 7.3864 and 7.0422.
In all three cases, the length of the time step used for simulation is $h=0.02$.

Here, we assess the precision of the algorithm using the relative MSE in order to compare all three parameter sets, which yield vastly different prices.
The relative quadratic error is approximated by computing the price estimators 20 times, for $N \in \{5000,10000,50000\}$ simulations.
The integral with respect to time (see step (3) of Algorithm \ref{algo:simulation}) is approximated using $M \in \{2,4\}$ sub-intervals and Simpson's $\frac 13$ rule.

The results of Table \ref{tab:RMSE123} show that the precision of the simulation algorithm seem to be affected by $n$.
Indeed, as a percentage of the exact price, the MSE of the price estimator is higher for PS1 than for the other parameter sets.
This observation becomes clearer as $N$ increases.

We recall that for PS3, $\frac{4\tilde\kappa\tilde\theta}{\tilde\varepsilon^2}$ is an integer, while this is not the case for PS2. 
It follows that for this latter parameter set, the weights $\hat L^{(j)}_T$ are all different, while they are all equal to 1 for PS3.
One could expect the estimator using uneven weights to show a worse performance due to the possible great variance of the weights.
However, in this case, both estimators show similar a performance; the algorithm does not seem to be affected by the use of uneven weights.

Finally, Table \ref{tab:RMSE123} shows that increasing $M$ may not significantly improve the precision of the price estimator.
Such an observation is important, since adding subintervals in the calculation of the time-integral slows down the algorithm. 
Keeping the number of subintervals low reduces computational complexity of our algorithm, making it more attractive.

\begin{table}[H]
	\begin{center}
		\caption{Relative MSE as a percentage of $\pi_0$.}
		\label{tab:RMSE123}
		\begin{tabular}{ccccccc}
			\hline 
			$N$ & \multicolumn{2}{c}{PS1} & \multicolumn{2}{c}{PS2} & \multicolumn{2}{c}{PS3} \\ 
			\cline{2-7}      
			~ & $M=2$ & $M=4$ & $M=2$ & $M=4$ & $M=2$ & $M=4$\\ 
			\hline          
			5000  & 0.271 & 0.316 & 0.183 & 0.225 & 0.239 & 0.214\\  
			10000 & 0.203 & 0.158 & 0.111 & 0.112 & 0.172 & 0.143\\	
			50000 & 0.158 & 0.135 & 0.085 & 0.083 & 0.067 & 0.070\\
			\hline
		\end{tabular}
	\end{center}
\end{table}

\subsubsection{Comparison to other algorithms}

In this section, we compare the performance of our new simulation algorithm for the 3/2 model to existing ones.
The first benchmark algorithm we consider is based on a Milstein-type discretization of the log-price and variance process.
The second one is based on the quadratic exponential scheme proposed by \cite{andersen2007efficient} as a modification to the method of \cite{BroadieKaya:2006}, which we adapted to the 3/2 model.
These algorithms are outlined in the appendix.

To assess the relative performance of the algorithms, we price in-the-money ($K/S_0 = 0.95$), at-the-money ($K/S_0 = 1$) and out-of-the-money ($K/S_0 = 1.05$) call options with $T=1$ year to maturity.
The exact prices of the options, which are used to calculate the MSE of the price estimates, are given in Table \ref{tab:prices}.
We consider all parameter sets with the exception of PS1, since this parametrization requires the simulation of 204 Ornstein-Uhlenbeck process, which makes our algorithm excessively slow. 
Run times for the calculation of the Monte Carlo estimators using $N=50,000$ simulations and $M=2$ subintervals are reported in Table 
\ref{tab:times}.

\begin{table}[H]
	\begin{center}
		\caption{Exact prices $\pi_0$ of European call options.}
		\label{tab:prices}
		\begin{tabular}{ccccc}
			\hline
			$K/S_0$   & $PS2$  & $PS3$ & $ PS4$  & $PS5$ \\
			\hline
			0.95     & 10.364 & 10.055 & 11.657 & 11.724    \\
			1        & 7.386 & 7.042 & 8.926 & 8.999    \\
			1.05     & 4.938 & 4.586 & 6.636 & 6.710    \\
			\hline      
		\end{tabular}
	\end{center}
\end{table} 

Figures \ref{fig:MSE24} and \ref{fig:MSE35} present the relative MSE of the price estimator as a function of the number of simulations. 
We note that the parametrizations considered in Figure \ref{fig:MSE24} are such that $\frac{4\tilde\kappa\tilde\theta}{\tilde\varepsilon^2} \notin \mathbb N$, while the opposite is true for Figure \ref{fig:MSE35}.

Overall, the precision of our weighted simulation algorithm is similar to that of the other two algorithms studied.
However, certain parameter sets result in more precise estimates.
Figure \ref{fig:MSE24} shows that the MSE is consistently larger with the weighted simulation algorithms than with the benchmark ones for PS2.
However, with PS4, the weighted algorithm performs as well as the other two algorithms, or better. 
We note that for PS2, $n=5$ while for PS4, $n=12$. 
It was observed in \cite{kouritzinbranching} in the case of the Heston model that as $n$ increases, the weighted simulation algorithm seems to perform better relatively to other algorithms.
This observation also seems to hold in the case of the 3/2 model.

For parametrizations that satisfy $\frac{4\tilde\kappa\tilde\theta}{\tilde\varepsilon^2} \in \mathbb N$, such as in Figure \ref{fig:MSE35}, we observe that the weighted simulation algorithm is at least as precise, and often more, than the other algorithms.
In this case, all the weights $\hat L_T$ are even, which tends to decrease the variance of the price estimator and thus, to decrease the relative MSE. 
It is also interesting to note that in the case of Figure \ref{fig:MSE35}, since $\frac{4\tilde\kappa\tilde\theta}{\tilde\varepsilon^2} \in \mathbb N$, it is not necessary to simulate $\tau^{(j)}_\delta$ and the trajectories $\hat L_T$.
Indeed, in this case, it is possible to simplify the algorithm using Proposition \ref{prop:solution_condition_C}, which tends to speed it up.

The run times presented in Table \ref{tab:times} show that in general, our method is slower than Milstein's, but faster than 
the quadratic exponential approximation of \cite{andersen2007efficient}. 
While the run times of the two benchmarks we consider are somewhat constant across the different parametrizations we tested, the speed of our method depends on a two factors; the number of Ornstein-Uhlenbeck processes $n$ to simulate and whether or not $\frac{4\tilde\kappa\tilde\theta}{\tilde\varepsilon^2} \in \mathbb N$.
This second factor explains the minor differences between the run times reported for PS2 and PS3. 
However, it should be noted that simulating the weights $\hat L$ is not particularly time consuming, as they are obtained as a deterministic function of $\hat U$ and therefore require no additional simulation.
The most significant difference in run times is due to $n$; for example, it takes twice as long to obtain a price estimate using PS5 ($n=12$,$\frac{4\tilde\kappa\tilde\theta}{\tilde\varepsilon^2} \notin \mathbb N$) than PS3 ($n=5$,$\frac{4\tilde\kappa\tilde\theta}{\tilde\varepsilon^2} \in \mathbb N$).
While our method is always faster than the one of \cite{andersen2007efficient} for the parametrizations studied, we expect that in certain cases (when $n$ is high), it could become slower.
Nonetheless, in those cases, our method should be very precise.

We also remark that, when it is used to simulate Heston prices and volatilities, Milstein's method can lead to poor accuracy, especially when the Feller condition is not respected.
In the 3/2 model, the Feller condition is always met, so it is normal to expect Milstein's algorithm to perform well.
Indeed, Figures \ref{fig:MSE24} and \ref{fig:MSE35} show that it reaches a similar level of precision as the other methods considered.

\begin{figure}[H]
	\centering	
	\begin{subfigure}[b]{0.32\textwidth}
		\includegraphics[width=\textwidth]{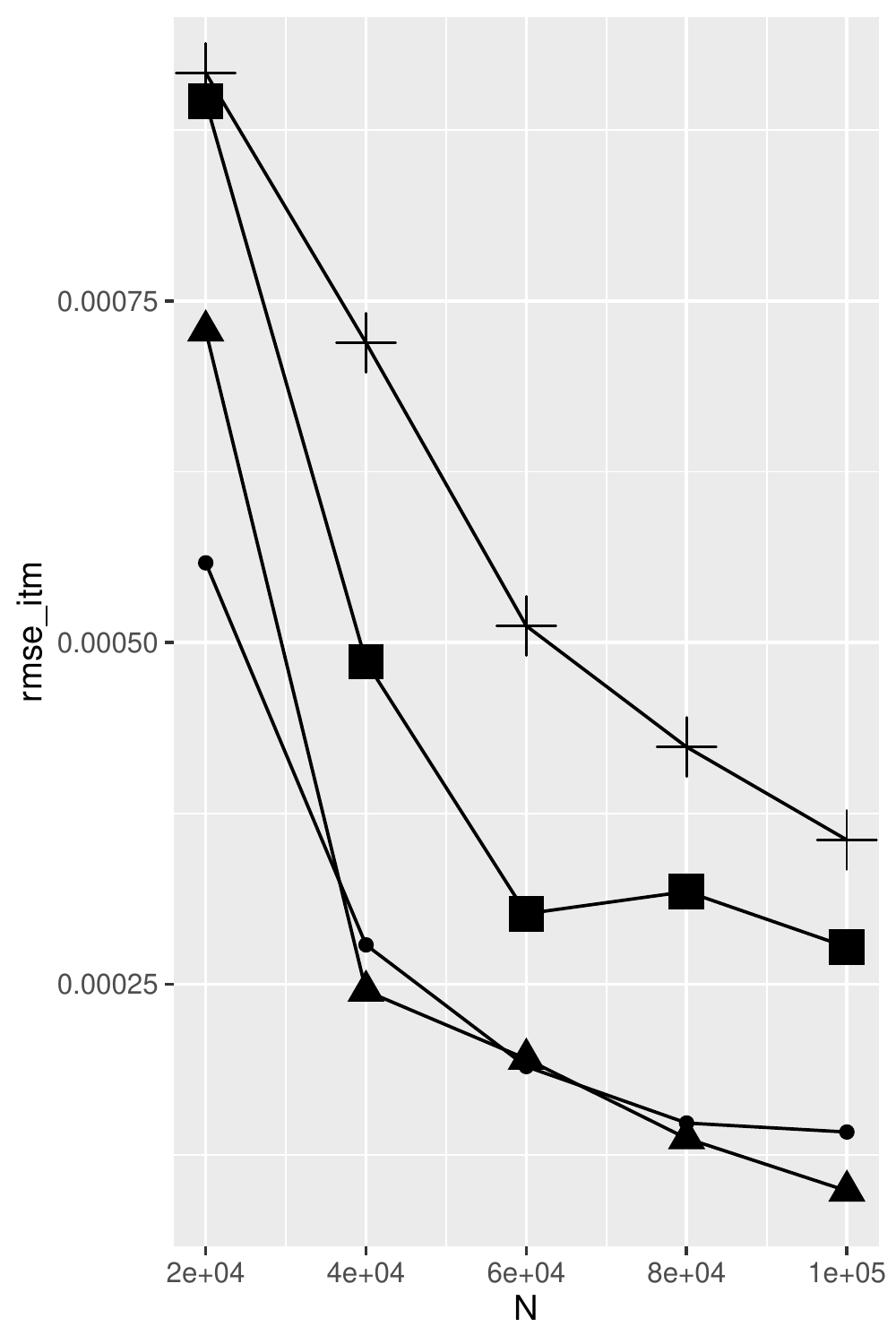}
		\subcaption{PS2, $K/S_0 = 0.95$}
		\label{fig:PS2itm}
	\end{subfigure}
	~
	\begin{subfigure}[b]{0.32\textwidth}
		\includegraphics[width=\textwidth]{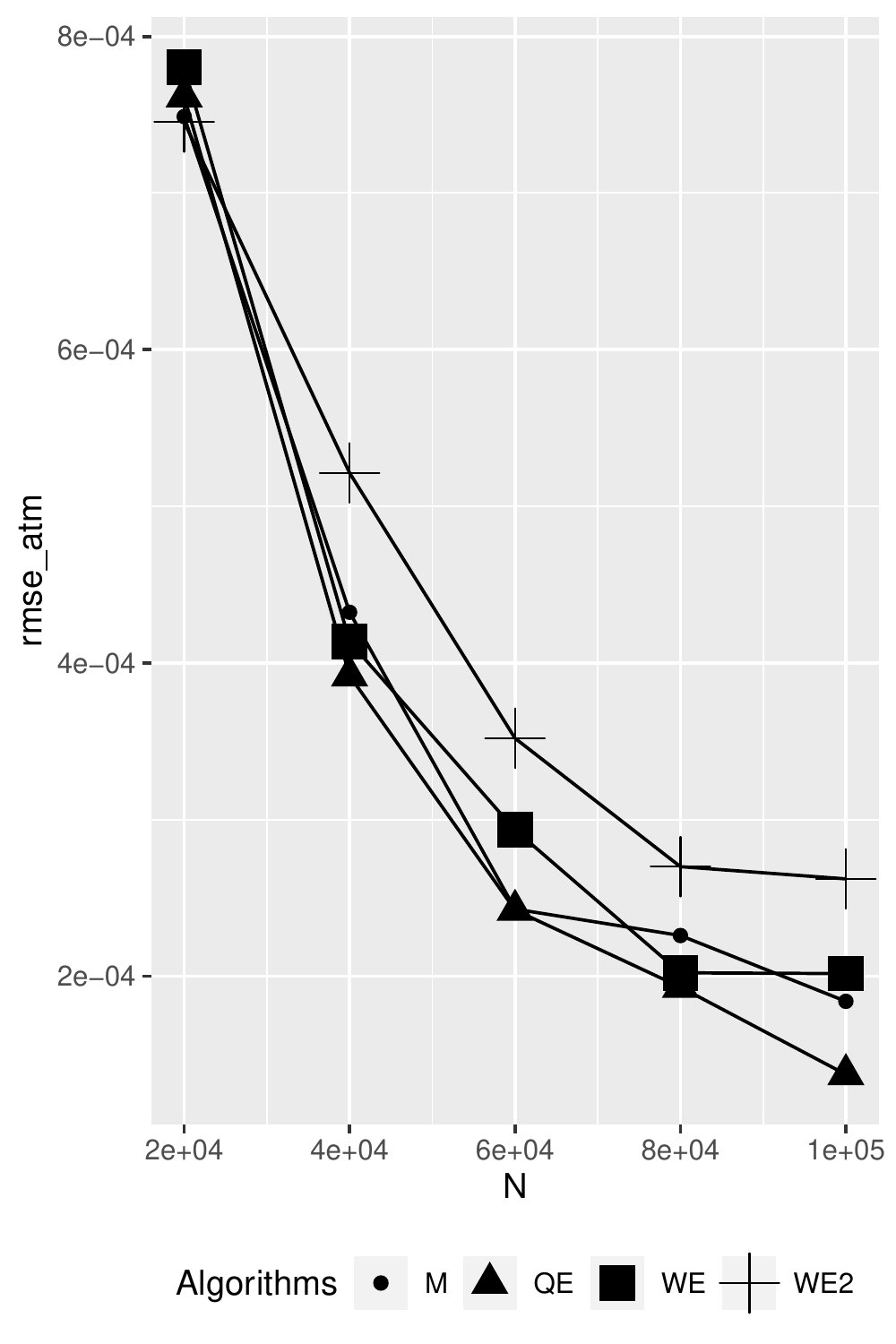}
		\subcaption{PS2, $K/S_0 = 1$}
		\label{fig:PS2atm}
	\end{subfigure}
	~
	\begin{subfigure}[b]{0.32\textwidth}
		\includegraphics[width=\textwidth]{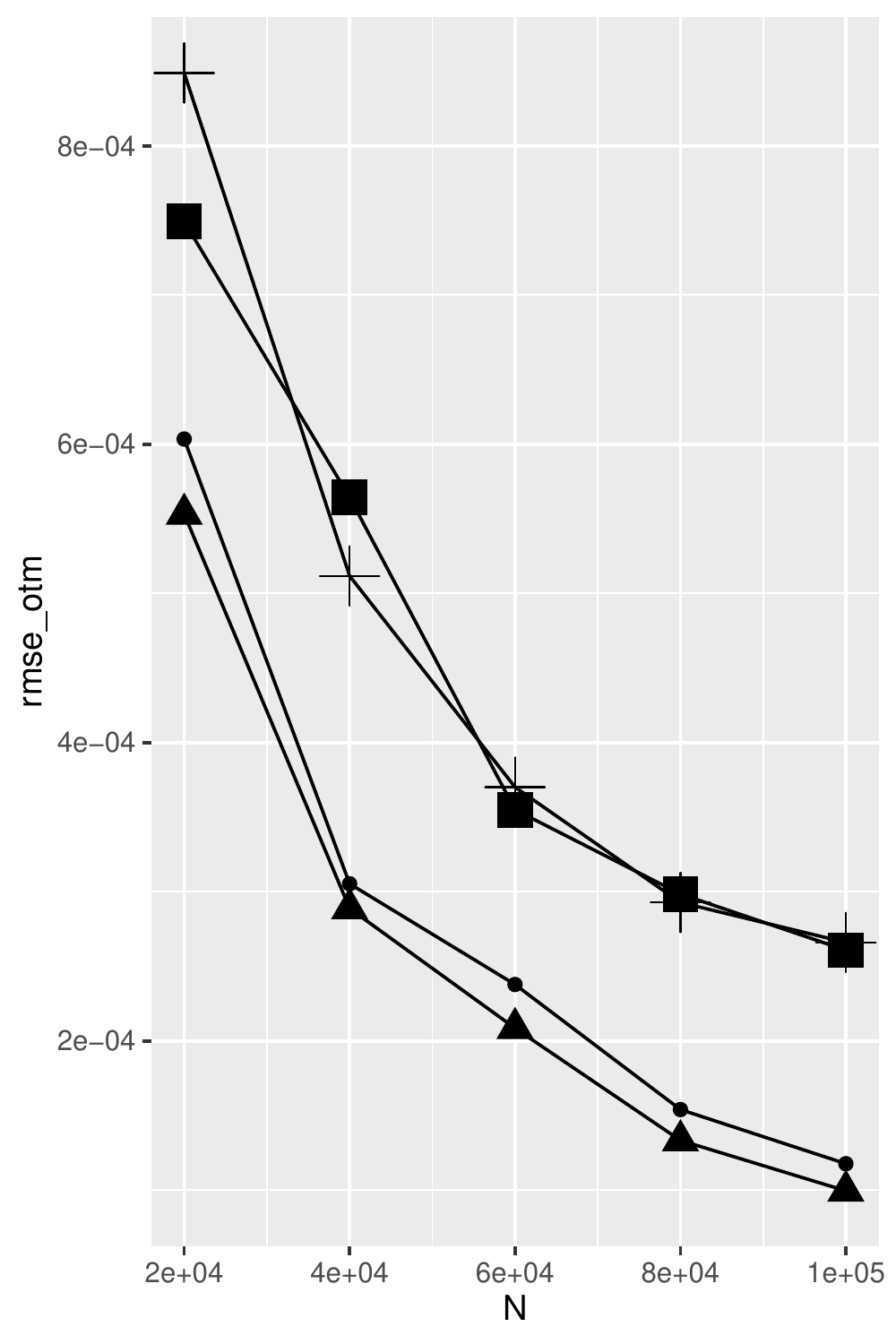}
		\subcaption{PS2, $K/S_0 = 1.05$}
		\label{fig:PS2otm}
	\end{subfigure}
	\\
	\begin{subfigure}[b]{0.32\textwidth}
		\includegraphics[width=\textwidth]{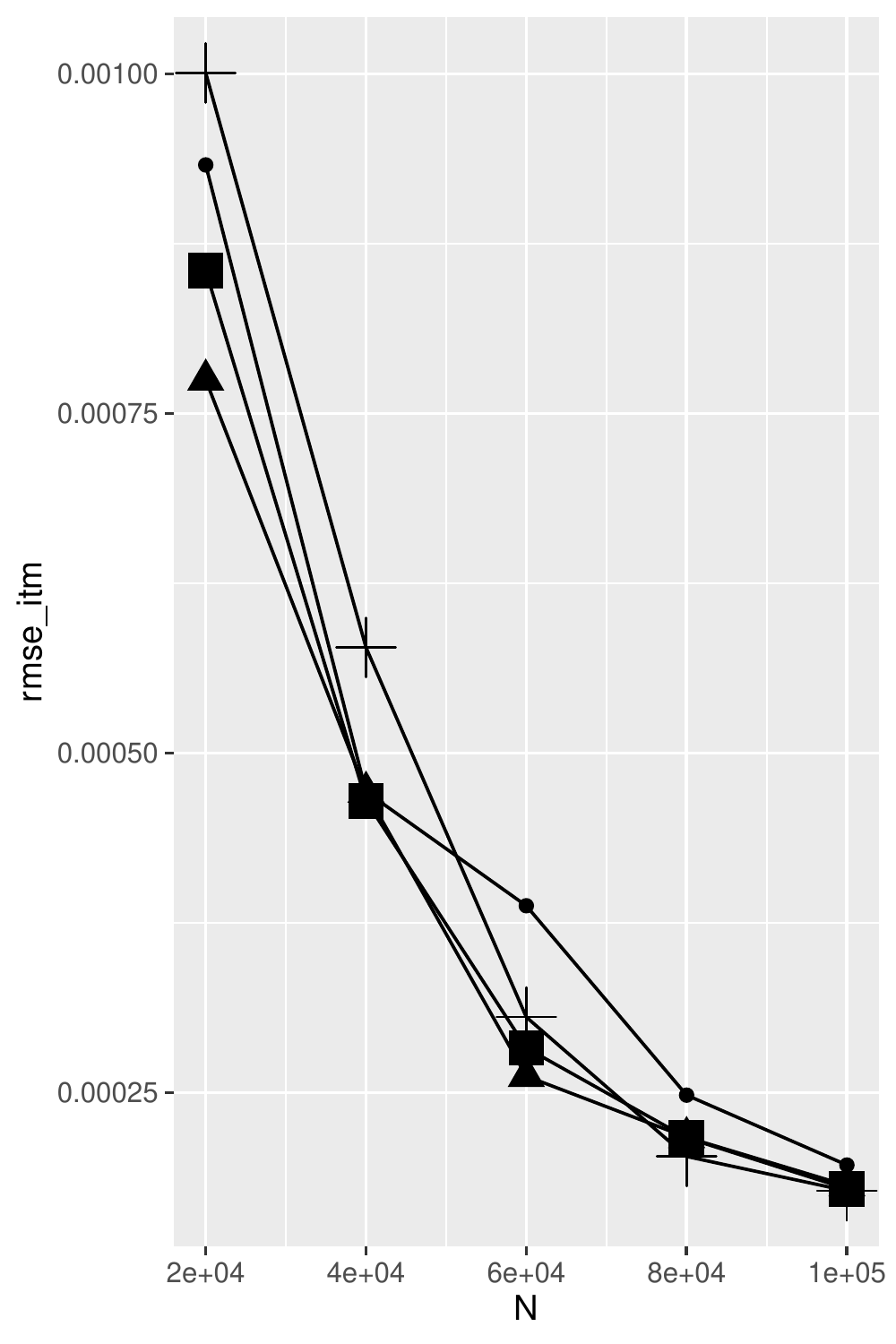}
		\subcaption{PS4, $K/S_0 = 0.95$}
		\label{fig:PS4itm}
	\end{subfigure}
	~
	\begin{subfigure}[b]{0.32\textwidth}
		\includegraphics[width=\textwidth]{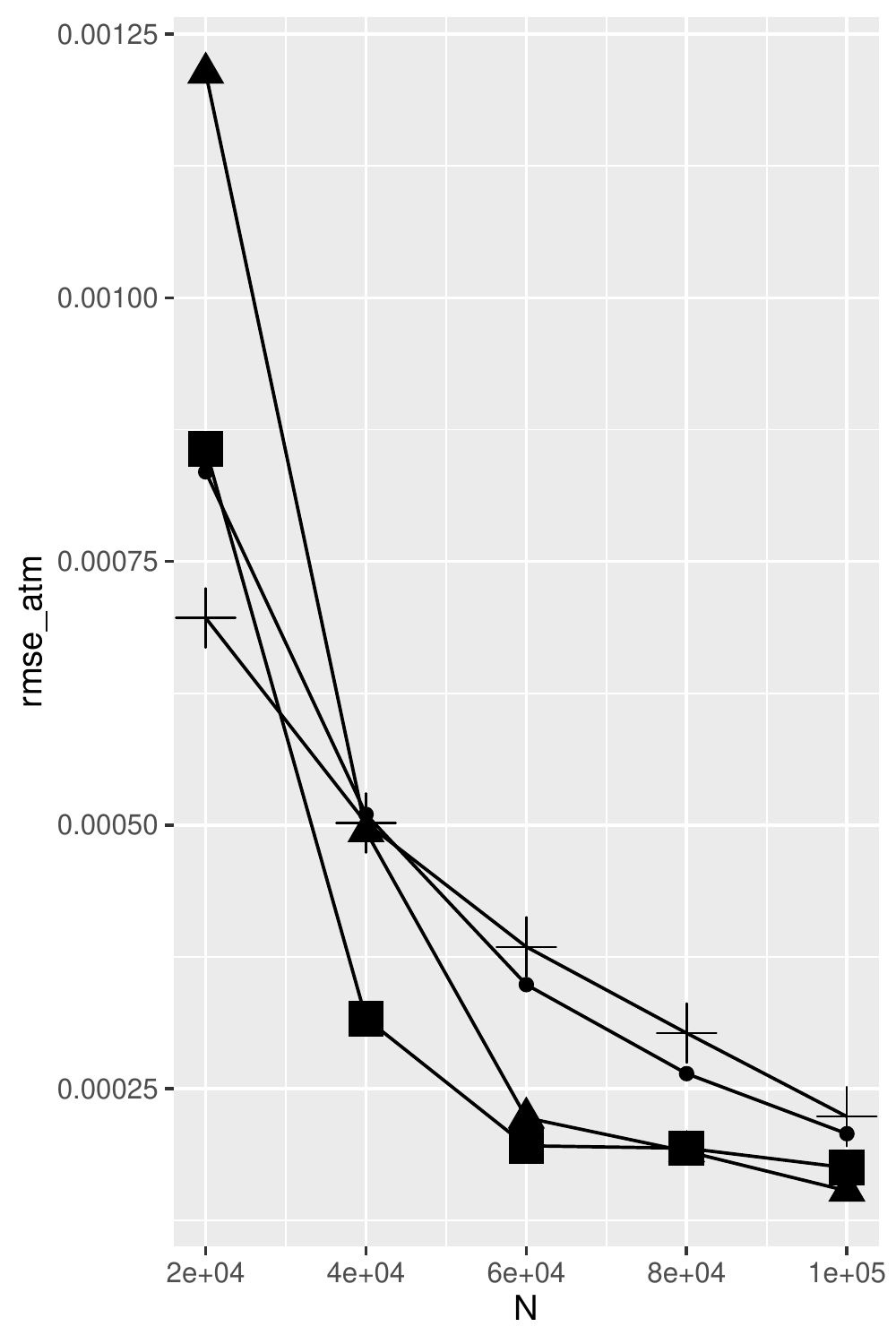}
		\subcaption{PS2, $K/S_0 = 1$}
		\label{fig:PS4atm}
	\end{subfigure}
	~
	\begin{subfigure}[b]{0.32\textwidth}
		\includegraphics[width=\textwidth]{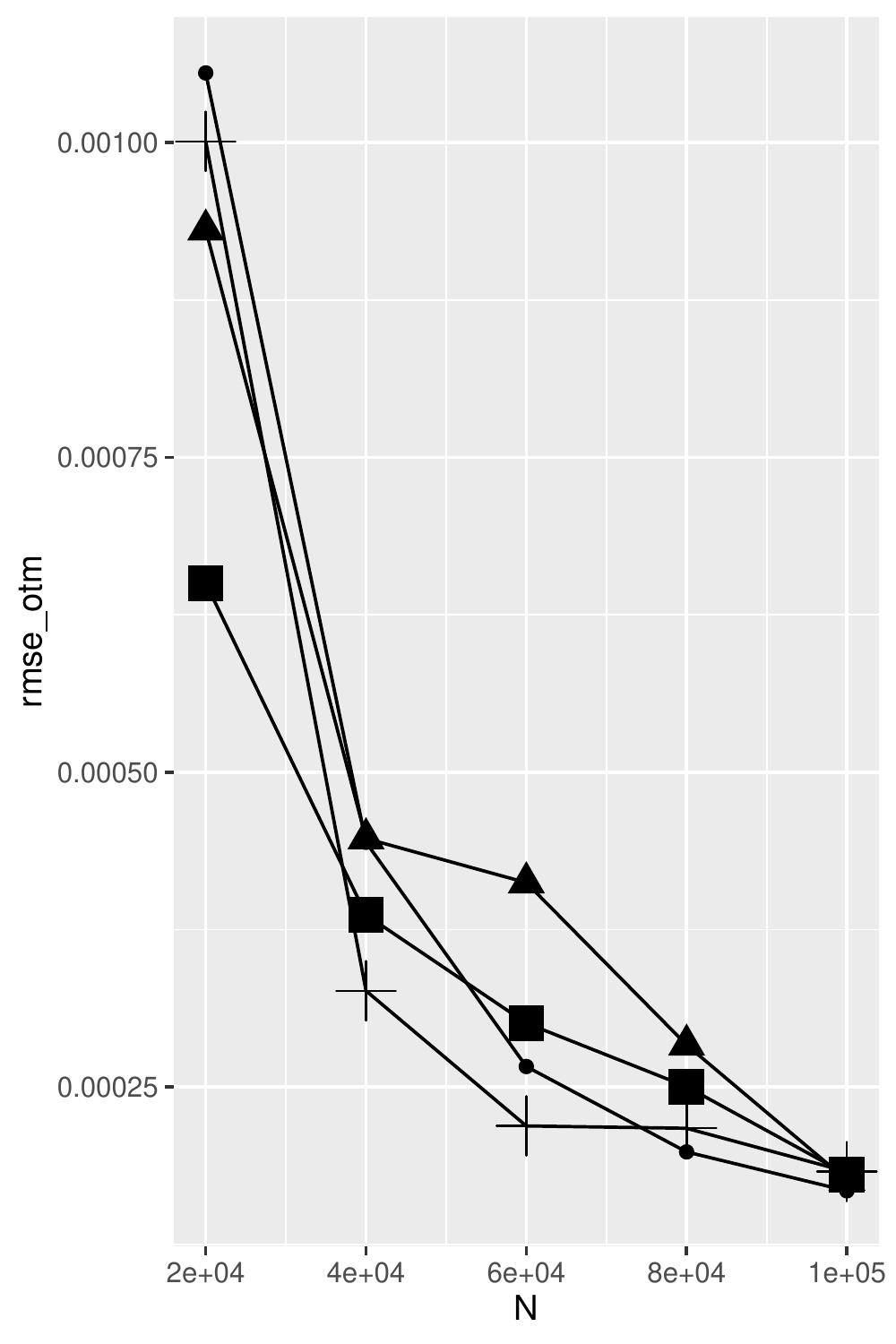}
		\subcaption{PS4, $K/S_0 = 1.05$}
		\label{fig:PS4otm}
	\end{subfigure}
	\caption{Relative MSE as a function of $N$, PS2 and PS4, algorithms: Milstein (dot), quadratic exponential (triangle), weighted, $M=2$ (square), weighted, $M=4$ (cross). }
	\label{fig:MSE24}
\end{figure}

\begin{figure}[H]
	\centering	
	\begin{subfigure}[b]{0.3\textwidth}
		\includegraphics[width=\textwidth]{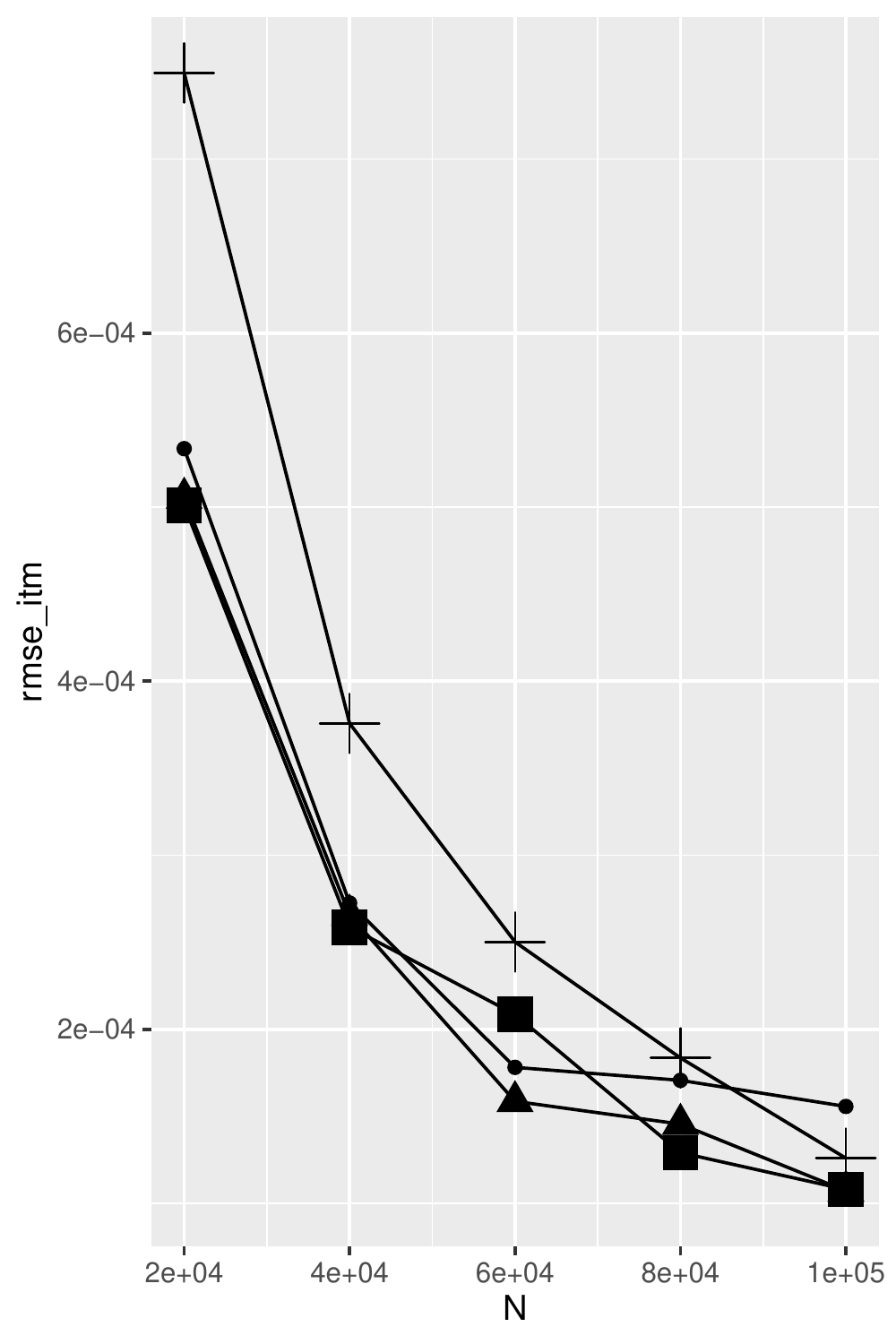}
		\subcaption{PS3, $K/S_0 = 0.95$}
		\label{fig:PS3itm}
	\end{subfigure}
	~
	\begin{subfigure}[b]{0.3\textwidth}
		\includegraphics[width=\textwidth]{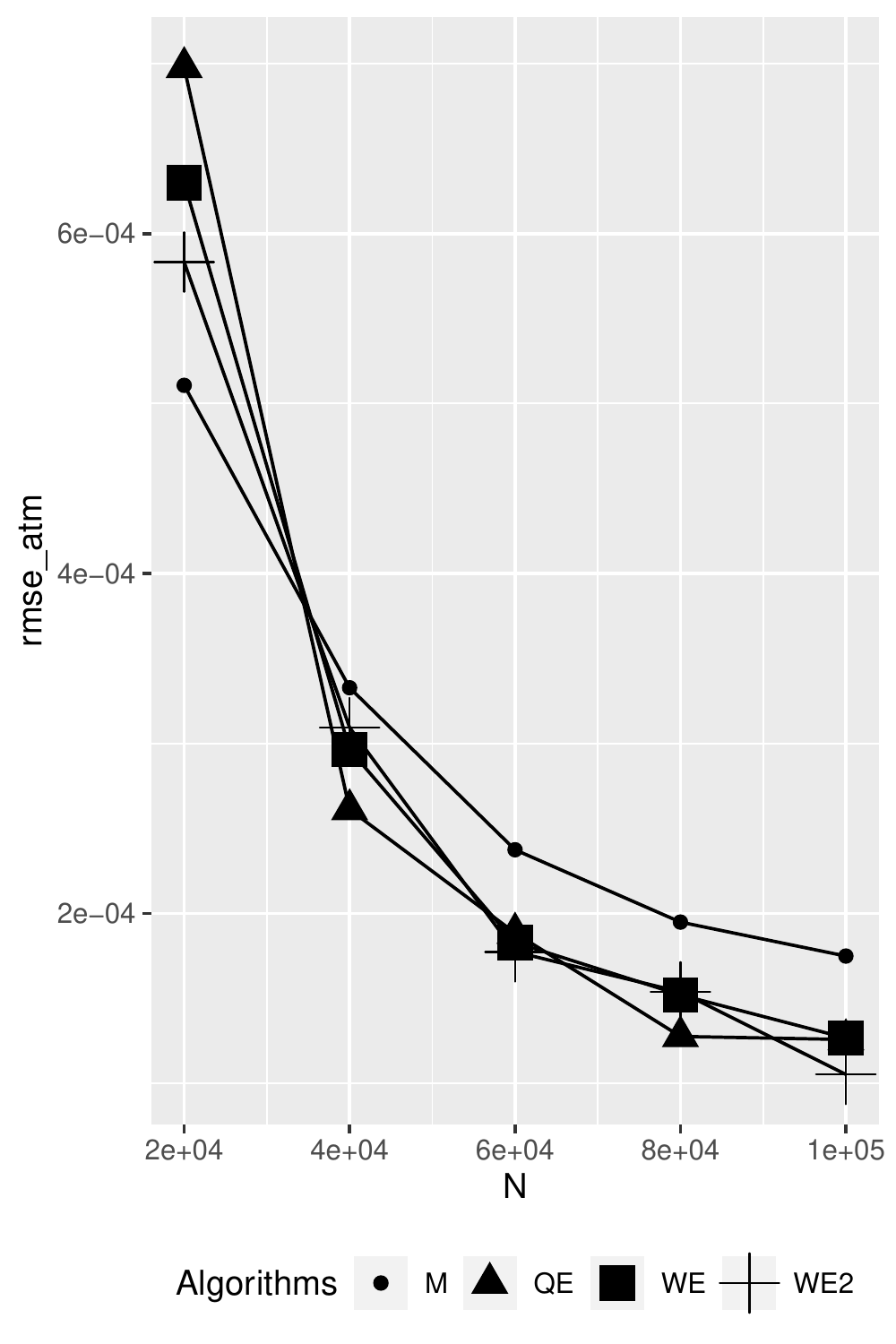}
		\subcaption{PS3, $K/S_0 = 1$}
		\label{fig:PS3atm}
	\end{subfigure}
	~
	\begin{subfigure}[b]{0.3\textwidth}
		\includegraphics[width=\textwidth]{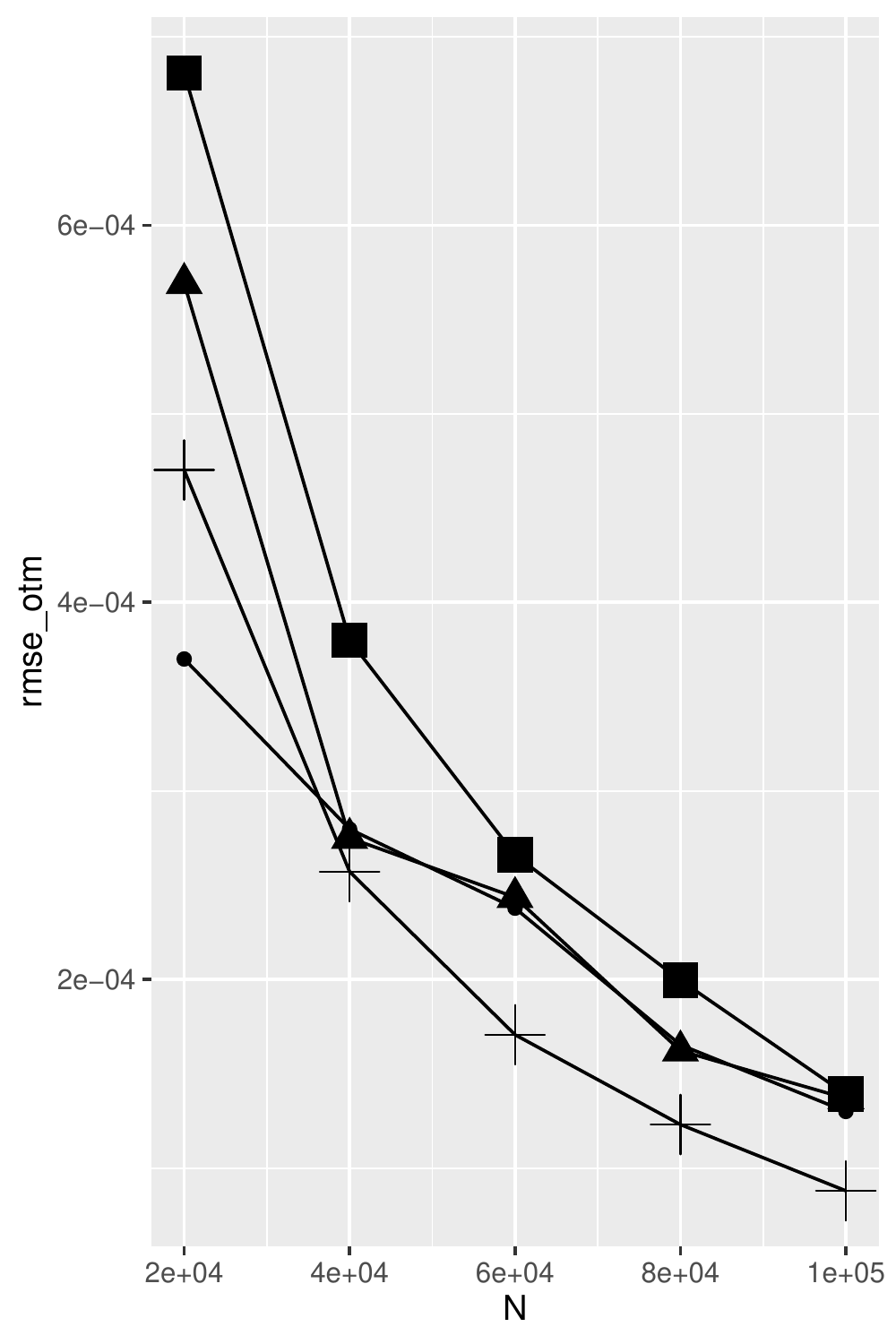}
		\subcaption{PS3, $K/S_0 = 1.05$}
		\label{fig:PS3otm}
	\end{subfigure}
	\\
	\begin{subfigure}[b]{0.3\textwidth}
		\includegraphics[width=\textwidth]{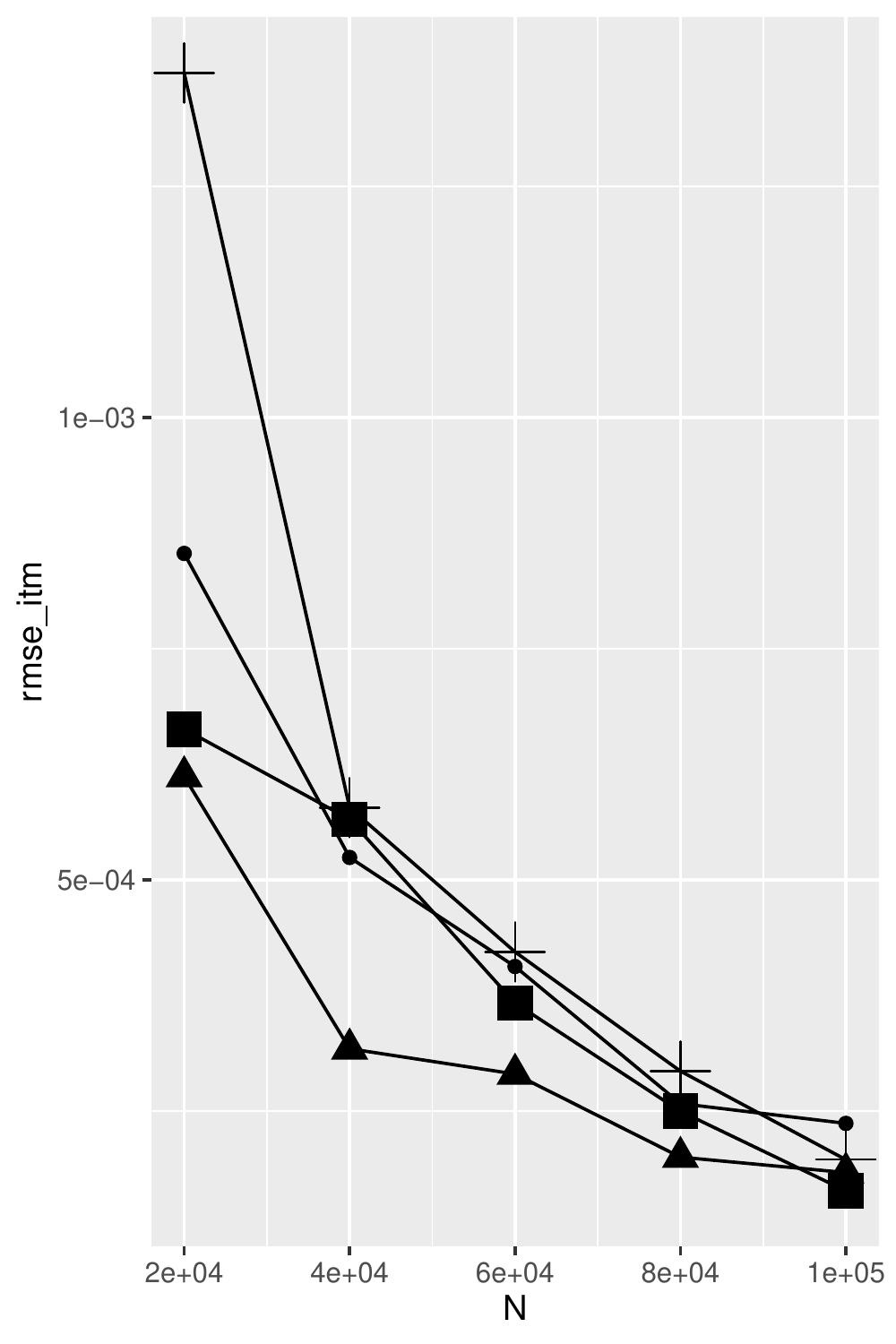}
		\subcaption{PS5, $K/S_0 = 0.95$}
		\label{fig:PS5itm}
	\end{subfigure}
	~
	\begin{subfigure}[b]{0.3\textwidth}
		\includegraphics[width=\textwidth]{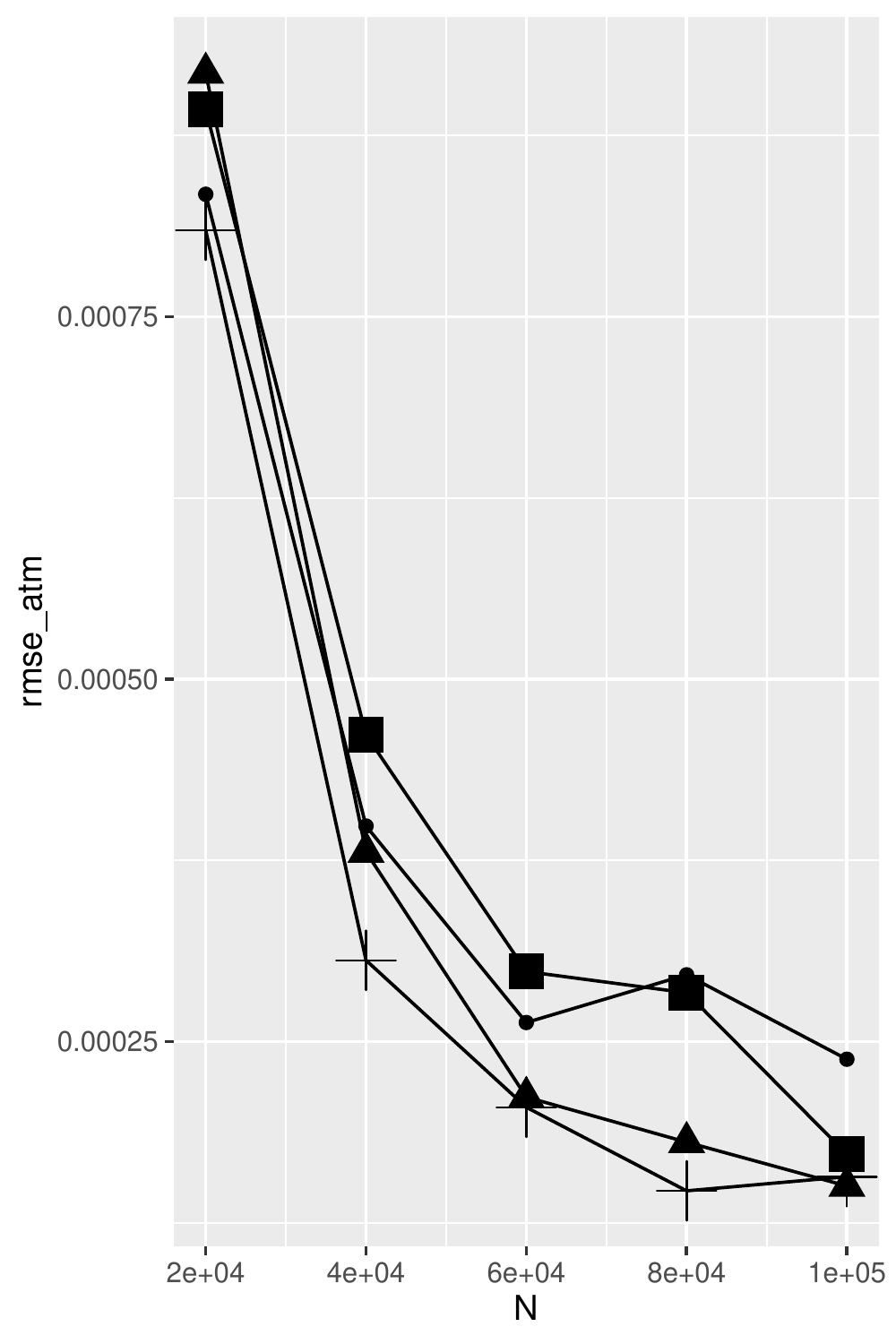}
		\subcaption{PS5, $K/S_0 = 1$}
		\label{fig:PS5atm}
	\end{subfigure}
	~
	\begin{subfigure}[b]{0.3\textwidth}
		\includegraphics[width=\textwidth]{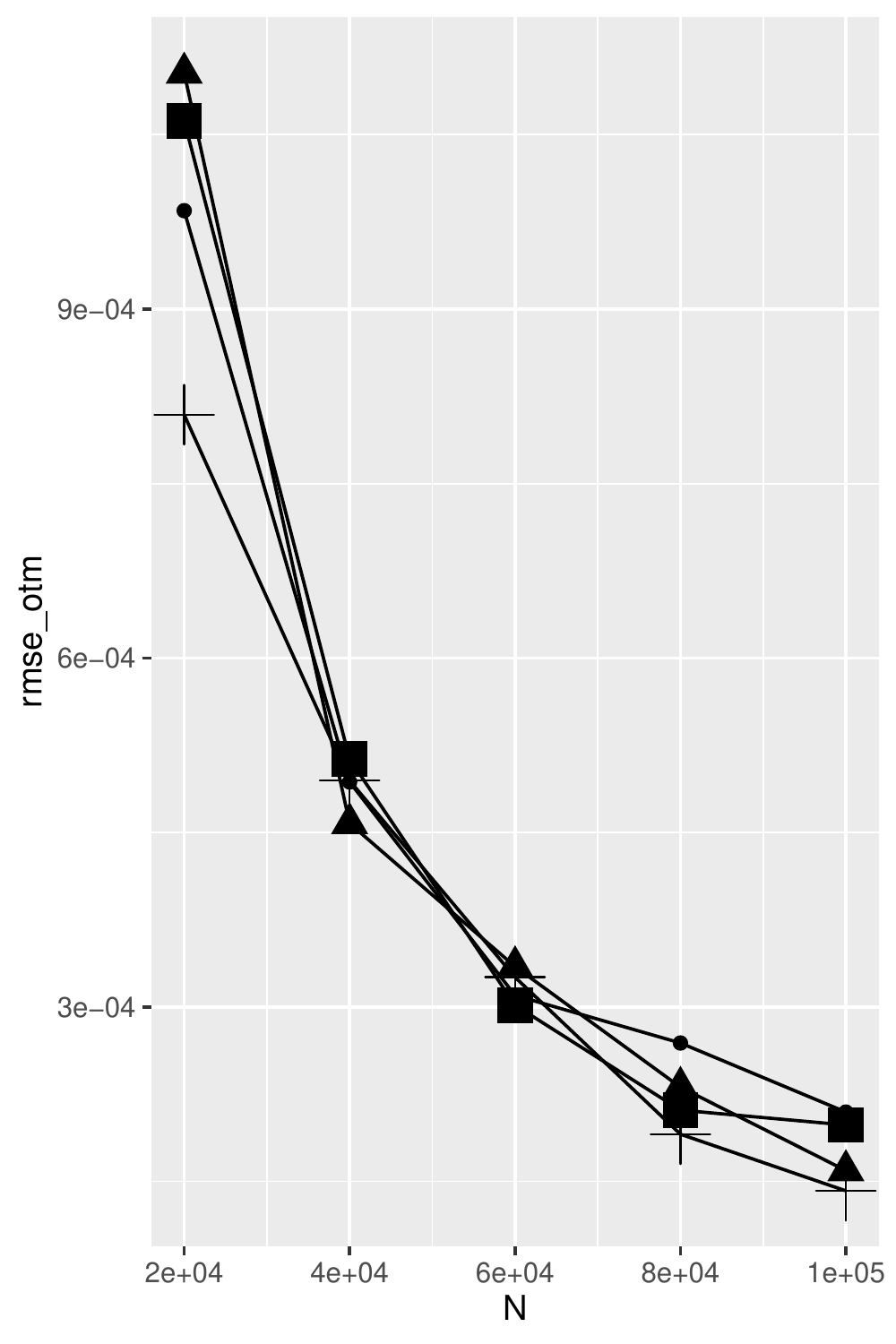}
		\subcaption{PS5, $K/S_0 = 1.05$}
		\label{fig:PS5otm}
	\end{subfigure}
	\caption{Relative MSE as a function of $N$, PS3 and PS5, algorithms: Milstein (dot), quadratic exponential (triangle), weighted, $M=2$ (square), weighted, $M=4$ (cross). }
	\label{fig:MSE35}
\end{figure}

\begin{table}[H]
	\begin{center}
		\caption{Run times (in seconds), $N=5 \times 10^5$, $M=2$.}
		\label{tab:times}
		\begin{tabular}{cccc}
			\hline
			\multirow{2}{*}{Parameters}   & 
			\multirow{2}{*}{Milstein}  & Quadratic  & \multirow{2}{*}{Weighted} \\
			 & & 
			 exponential & \\
			\hline
			PS2  
			& 0.445 & 2.222 & 0.712 \\
			PS3  
			& 0.446 & 2.224 & 0.615 \\
			PS4  
			& 0.444 & 2.226 & 1.314 \\
			PS5  
			& 0.444 & 2.222 & 1.312 \\
			\hline      
		\end{tabular}
	\end{center}
\end{table}

\section{Conclusion}

In this paper, we present a weak explicit solution to the 3/2 model, up until the inverse of the variance process drops below a given threshold.
We develop a simulation algorithm based on this solution and show that it can be used to price options in the 3/2 model, since in practice, the inverse variance process stays away from 0.
Numerical examples show that our simulation algorithm performs at least as well as popular algorithms presented in the literature. Precision is improved when the parameters satisfy $\frac{4\tilde\kappa\tilde\theta}{\tilde\varepsilon^2} \in \mathbb N$ and when $n$ is larger.
We also show that it is significantly faster than the quadratic exponential approximation of \cite{andersen2007efficient} to the method of \cite{BroadieKaya:2006}, which is generally considered to present a good balance between precision and computation time.

It is important to note that the method that we present in this paper could be significantly sped up by the use of sequential resampling, as implemented in \cite{kouritzinbranching} for the Heston model.
Such improvements, left for future work, could give a significant advantage to our weighted simulation algorithm for the 3/2 model.

\section*{Appendix}

This section presents the simulation algorithms used to produce the numerical examples in Section \ref{sec:NumericalResults}.
Algorithm \ref{algo:simulation} stems from the results we present Theorem \ref{thm:solution}.
Algorithm \ref{algo:milstein} is a Milstein-type algorithm applied to the 3/2 model.
Algorithm \ref{algo:andersen} is \cite{andersen2007efficient}'s approximation to the algorithm proposed by \cite{BroadieKaya:2006}, modified for the 3/2 model, since the original algorithm was developed for the Heston model.
Algorithms \ref{algo:milstein} and \ref{algo:andersen} are considered for comparison purposes.

For all algorithms, we consider a partition $\{0,h,2h,\ldots,mh\}$, with $mh=T$ of the time interval $[0,T]$, and outline the simulation of $N$ paths of $(\hat S, \hat U, \hat L)$, as well as the associated stopping times $\tau_\delta$.

To simplify the exposition of Algorithm \ref{algo:simulation}, we define the following constants:
\begin{align*}
\alpha_h = e^{-\frac{\tilde\kappa}{2}h}, \qquad\qquad 
\sigma_h = \frac{\tilde\varepsilon^2}{4\tilde\kappa}\left(1-e^{-h\tilde\kappa}\right).
\end{align*}
We also drop the hats to simplify the notation.

\begin{myalg}[Weighted explicit simulation]\label{algo:simulation}\ \\
		\textbf{I. Initialize:}	
		\begin{itemize}
		\item[]
		Set the starting values for each simulated path:
		\begin{align*}
			\{(S_0^{(j)},L_0^{(j)},\tau_\delta^{(j)}) = (S_0,1,T+h)\}_{j=1}^N, \,
			\{Y_0^{(l,j)} = \sqrt{U_0/n}\}_{l,j=1}^{n,N}
		\end{align*}
		\end{itemize}
		\smallskip
		
		\noindent
		\textbf{II. Loop on time:} for $i= 1,\ldots, m$
		\begin{itemize}
			\item[]
			\textbf{Loop on particles:} for $j = 1,\ldots,N$, do
			\begin{enumerate}[label=(\arabic*)]
				\setlength\itemsep{7pt}
				\item 
				For $l = 1,\ldots, n$, generate $Y^{(l,j)}_{ih}$ using
				$Y^{(l,j)}_{ih} \sim N\left(\alpha_h Y^{(l,j)}_{(i-1)h},\sigma^2_h\right).$
				
				\item
				Set 
				$U^{(j)}_{ih} = \sum_{l=1}^n (Y^{(l,j)}_{ih})^2$.
				
				\item
				Let $IntU^{(j)}_{ih} \approx \int_{(i-1)h}^{ih}  (U_s^{(j)})^{-1} \d s$ using \eqref{eq:int_U} (or another quadrature rule).
				
				\item
				Generate $S^{(j)}_{ih}$ from $S^{(j)}_{(i-1)h}$ using \eqref{eq:hat_S}, with ${\int_{ih}^{(i-1)h} \hat (U_s^{(j)})^{-1/2} dW^{(2)}_s \sim N(0,IntU^{(j)}_{ih})}$.
				
				\item
				If $ih \leq \tau^{(j)}_\delta$,
				\begin{enumerate}[label=(\roman*)]
					\item
					If $U^{(j)}_{ih} > \delta$, generate $L^{(j)}_{ih}$ from $L^{(j)}_{(i-1)h}$ using \eqref{eq:L_th}.
					\item
					Otherwise, set $\tau^{(j)}_\delta = t$.
				\end{enumerate}
			\end{enumerate}
		\end{itemize}	
\end{myalg}

\begin{myalg}[Milstein]\label{algo:milstein}\ \\
	\textbf{I. Initialize:}	
	\begin{itemize}
		\item[]
		Set the starting values for each simulated path:
		\begin{align*}
		\{(S_0^{(j)},U_0^{(j)}) = (S_0,U_0\}_{j=1}^N
		\end{align*}
	\end{itemize}
	\smallskip
	
	\noindent
	\textbf{II. Loop on time:} for $i= 1,\ldots, m$
	\begin{itemize}
		\item[]
		\textbf{Loop on particles:} for $j = 1,\ldots,N$, do
		\begin{enumerate}[label=(\arabic*)]
			\item
			Correct for possible negative values:
			$\bar{u}^{(j)} = \max(U((i-1)h),0)$
			\item
			Generate $U^{(j)}_{ih}$ from $U^{(j)}_{(i-1)h}$:
			\begin{align*}
			U^{(j)}_{ih} &= U^{(j)}_{(i-1)h} + \tilde\kappa (\tilde \theta - \bar u^{(j)}) h \\
			&\qquad + \tilde \epsilon \sqrt{\bar u^{(j)} h} Z^{(j)}_1 + \frac 14 \tilde \varepsilon^2 ((Z^{(j)}_1)^2 -1)h,
			\end{align*}
			with $Z^{(j)}_1 \sim N(0,1)$.
			\item
			Generate $S^{(j)}_{ih}$ from $S^{(j)}_{(i-1)h}$:
			\begin{align*}
			S_{ih} = S_{(i-1)h} \exp\left\{\left(r-\frac 1{2\bar u}\right)h + \sqrt{\frac h {\bar{u}}}Z^{(j)}_2\right\},
			\end{align*} 
			with $Z^{(j)}_2 \sim N(0,1)$.
		\end{enumerate}
	\end{itemize}
\end{myalg}

\begin{myalg}[Quadratic exponential]\label{algo:andersen}\ \\
	\textbf{I. Initialize:}	
\begin{enumerate}[label=(\arabic*)]
	\item
	Set the starting values for each simulated path:
	\begin{align*}
	\{(S_0^{(j)},U_0^{(j)}) = (S_0,U_0)\}_{j=1}^N
	\end{align*}
	\item
	Fix the constant $\phi_c \in [1,2]$.
\end{enumerate}
\smallskip

\noindent
\textbf{II. Loop on time:} for $i= 1,\ldots, m$
\begin{itemize}
	\item[]
	\textbf{Loop on particles:} for $j = 1,\ldots,N$, do
	\begin{enumerate}[label=(\arabic*)]
		\item
		Set the variables $m_{i,j}$ and $s_{i,j}$:
		\begin{align*}
		m_{i,j} &= \tilde\theta + (U^{(j)}_{(i-1)h}-\tilde\theta)e^{-\tilde\kappa h}\\
		s_{i,j} &= \frac{U_{(i-1)h}\tilde\varepsilon^2 e^{-\tilde\kappa h}}{\tilde\kappa}(1-e^{-\tilde\kappa h})
		+ \frac{\tilde\theta \tilde\varepsilon^2} {2\tilde\kappa}(1-e^{-\tilde\kappa h})^2
		\end{align*} 
		\item
		Set $\phi_{i,j} = \frac{s^2_{i,j}}{m^2_{i,j}}$.
		\item
		If $\phi_{i,j} < \phi_c$,
			\item[] 
			Generate $U^{(j)}_{ih}$ from $U^{(j)}_{(i-1)h}$:
			\begin{align*}
			U^{(j)}_{ih} = a_{i,j}(b_{i,j}+Z^{(j)})^2,
			\end{align*} 
			where $Z^{(j)} \sim N(0,1)$ and 
			\begin{align*}
			b^2_{i,j} &= 2\phi_{i,j}^{-1} -1+ \sqrt{2\phi_{i,j}^{-1}}\sqrt{2\phi_{i,j}^{-1}-1}\\
			a_{i,j} &= \frac{m_{i,j}}{1+b^2_{i,j}}
			\end{align*}
		\item
		If $\phi_{i,j} \geq \phi_c$,
			\item[]
			Generate $U^{(j)}_{ih}$ from $U^{(j)}_{(i-1)h}$:
			\begin{align*}
			U^{(j)}_{ih} = \frac 1\beta \log \left(\frac{1-p_{i,j}}{1-X^{(j)}}\right),
			\end{align*} 
			where $X^{(j)} \sim Uniform(0,1)$ and 
			\begin{align*}
			p_{i,j} = \frac{\psi_{i,j}-1}{\psi_{i,j}+1}, \qquad
			\beta_{i,j} = \frac{1-p_{i,j}}{m_{i,j}}
			\end{align*}
		\item
		Let $IntU^{(j)}_{ih} \approx \int_{(i-1)h}^{ih}  (U_s^{(j)})^{-1} \d s$ using \eqref{eq:int_U}.
		
		\item
		Generate $S^{(j)}_{ih}$ from $S^{(j)}_{(i-1)h}$ using \eqref{eq:hat_S}, with ${\int_{ih}^{(i-1)h} \hat (U_s^{(j)})^{-1/2} dW^{(2)}_s \sim N(0,IntU^{(j)}_{ih})}$.
		
	\end{enumerate}
\end{itemize}	
\end{myalg}

\bibliographystyle{abbrvnat}
\bibliography{Simulation_32}

\begin{thebibliography}{22}
\providecommand{\natexlab}[1]{#1}
\providecommand{\url}[1]{\texttt{#1}}
\expandafter\ifx\csname urlstyle\endcsname\relax
  \providecommand{\doi}[1]{doi: #1}\else
  \providecommand{\doi}{doi: \begingroup \urlstyle{rm}\Url}\fi

\bibitem[Ahn and Gao(1999)]{ahn1999parametric}
D.-H. Ahn and B.~Gao.
\newblock A parametric nonlinear model of term structure dynamics.
\newblock \emph{The Review of Financial Studies}, 12\penalty0 (4):\penalty0
  721--762, 1999.

\bibitem[Andersen(2007)]{andersen2007efficient}
L.~B. Andersen.
\newblock {Efficient simulation of the Heston stochastic volatility model}.
\newblock \emph{Journal of computational finance}, 11\penalty0 (3):\penalty0
  1--42, 2007.

\bibitem[Bakshi et~al.(2006)Bakshi, Ju, and Ou-Yang]{bakshi2006estimation}
G.~Bakshi, N.~Ju, and H.~Ou-Yang.
\newblock Estimation of continuous-time models with an application to equity
  volatility dynamics.
\newblock \emph{Journal of Financial Economics}, 82\penalty0 (1):\penalty0
  227--249, 2006.

\bibitem[Baldeaux(2012)]{baldeaux2012exact}
J.~Baldeaux.
\newblock Exact simulation of the 3/2 model.
\newblock \emph{International Journal of Theoretical and Applied Finance},
  15\penalty0 (05):\penalty0 1250032, 2012.

\bibitem[B{\'e}gin et~al.(2015)B{\'e}gin, B{\'e}dard, and
  Gaillardetz]{begin2015simulating}
J.-F. B{\'e}gin, M.~B{\'e}dard, and P.~Gaillardetz.
\newblock Simulating from the heston model: A gamma approximation scheme.
\newblock \emph{Monte Carlo Methods and Applications}, 21\penalty0
  (3):\penalty0 205--231, 2015.

\bibitem[Blount and Kouritzin(2010)]{blount2010convergence}
D.~Blount and M.~A. Kouritzin.
\newblock On convergence determining and separating classes of functions.
\newblock \emph{Stochastic processes and their applications}, 120\penalty0
  (10):\penalty0 1898--1907, 2010.

\bibitem[Broadie and Kaya(2006)]{BroadieKaya:2006}
M.~Broadie and O.~Kaya.
\newblock Exact simulation of stochastic volatility and other affine jump
  diffusion processes.
\newblock \emph{Operation Research}, 54\penalty0 (2):\penalty0 217--231, 2006.

\bibitem[Carr and Sun(2007)]{carr2007new}
P.~Carr and J.~Sun.
\newblock A new approach for option pricing under stochastic volatility.
\newblock \emph{Review of Derivatives Research}, 10\penalty0 (2):\penalty0
  87--150, 2007.

\bibitem[Chan and Platen(2015)]{chan2015pricing}
L.~Chan and E.~Platen.
\newblock Pricing and hedging of long dated variance swaps under a 3/2
  volatility model.
\newblock \emph{Journal of Computational and Applied Mathematics},
  278:\penalty0 181--196, 2015.

\bibitem[Drimus(2012)]{drimus2012options}
G.~G. Drimus.
\newblock Options on realized variance by transform methods: a non-affine
  stochastic volatility model.
\newblock \emph{Quantitative Finance}, 12\penalty0 (11):\penalty0 1679--1694,
  2012.

\bibitem[Ethier and Kurtz()]{ethiermarkov}
S.~Ethier and T.~G. Kurtz.
\newblock Markov processes: Characterization and convergence, 2005.

\bibitem[Goard and Mazur(2013)]{goard2013stochastic}
J.~Goard and M.~Mazur.
\newblock Stochastic volatility models and the pricing of vix options.
\newblock \emph{Mathematical Finance: An International Journal of Mathematics,
  Statistics and Financial Economics}, 23\penalty0 (3):\penalty0 439--458,
  2013.

\bibitem[Grasselli(2017)]{grasselli20174}
M.~Grasselli.
\newblock The 4/2 stochastic volatility model: a unified approach for the
  heston and the 3/2 model.
\newblock \emph{Mathematical Finance}, 27\penalty0 (4):\penalty0 1013--1034,
  2017.

\bibitem[Hanson(2010)]{hanson2010stochastic}
F.~B. Hanson.
\newblock {Stochastic calculus of Heston’s stochastic-volatility model}.
\newblock In \emph{Proceedings of the 19th International Symposium on
  Mathematical Theory of Networks and Systems--MTNS}, volume~5, 2010.

\bibitem[Heston(1997)]{heston1997simple}
S.~L. Heston.
\newblock A simple new formula for options with stochastic volatility.
\newblock 1997.

\bibitem[Itkin and Carr(2010)]{itkin2010pricing}
A.~Itkin and P.~Carr.
\newblock Pricing swaps and options on quadratic variation under stochastic
  time change models—discrete observations case.
\newblock \emph{Review of Derivatives Research}, 13\penalty0 (2):\penalty0
  141--176, 2010.

\bibitem[Kouritzin(2018)]{Kouritzin16}
M.~A. Kouritzin.
\newblock {Explicit Heston solutions and stochastic approximation for
  path-dependent option pricing}.
\newblock \emph{International Journal of Theoretical and Applied Finance},
  21\penalty0 (01):\penalty0 1850006, 2018.

\bibitem[Kouritzin and MacKay(2020)]{kouritzinbranching}
M.~A. Kouritzin and A.~MacKay.
\newblock {Branching particle pricers with Heston examples}.
\newblock \emph{International Journal of Theoretical and Applied Finance},
  2020.

\bibitem[Lewis(2000)]{lewis2000option}
A.~Lewis.
\newblock Option valuation under stochastic volatility with mathematica code,
  2000.

\bibitem[Lord et~al.(2010)Lord, Koekkoek, and Dijk]{lord2010comparison}
R.~Lord, R.~Koekkoek, and D.~V. Dijk.
\newblock A comparison of biased simulation schemes for stochastic volatility
  models.
\newblock \emph{Quantitative Finance}, 10\penalty0 (2):\penalty0 177--194,
  2010.

\bibitem[Yuen et~al.(2015)Yuen, Zheng, and Kwok]{yuen2015pricing}
C.~H. Yuen, W.~Zheng, and Y.~K. Kwok.
\newblock Pricing exotic discrete variance swaps under the 3/2-stochastic
  volatility models.
\newblock \emph{Applied Mathematical Finance}, 22\penalty0 (5):\penalty0
  421--449, 2015.

\bibitem[Zheng and Zeng(2016)]{zheng2016pricing}
W.~Zheng and P.~Zeng.
\newblock Pricing timer options and variance derivatives with closed-form
  partial transform under the 3/2 model.
\newblock \emph{Applied mathematical finance}, 23\penalty0 (5):\penalty0
  344--373, 2016.

\end{thebibliography}

\end{document}